\documentclass[3p,times]{elsarticle}

\usepackage[T1]{fontenc}
\usepackage{microtype}

\usepackage{amsmath,amssymb,amsfonts,amsthm}
\usepackage{nicefrac}
\usepackage{bm}

\usepackage{graphicx}
\usepackage{booktabs}
\usepackage{multirow}
\usepackage{array}
\usepackage{colortbl}
\usepackage{wrapfig}
\usepackage{graphbox}
\usepackage{subcaption}
\usepackage[textfont=it,labelfont=bf]{caption}

\usepackage{enumitem}
\usepackage{url}

\usepackage{algorithm}
\usepackage{algpseudocode}

\usepackage{mdframed}
\usepackage{xcolor}
\usepackage{shadethm}
\usepackage{thmtools}

\usepackage{pgfplots}
\pgfplotsset{compat=1.17}

\usepackage{hyperref}
\hypersetup{
    breaklinks=true,
    colorlinks=true,
    linkcolor=blue!80!black,
    citecolor=blue!80!black,
    urlcolor=blue!80!black,
    plainpages=true
}

\newcommand{\R}{\mathbb{R}}
\newcommand{\distH}{d_H}
\newcommand{\op}{\mathrm{op}}
\newcommand{\dist}{\operatorname{dist}}
\DeclareMathOperator{\Lip}{Lip}

\newcommand{\secref}[1]{\hyperref[#1]{Section~\ref*{#1}}}
\newcommand{\figref}[1]{\hyperref[#1]{Figure~\ref*{#1}}}
\newcommand{\tabref}[1]{\hyperref[#1]{Table~\ref*{#1}}}
\newcommand{\eqnref}[1]{\hyperref[#1]{(Equation~\ref*{#1})}}
\newcommand{\algoref}[1]{\hyperref[#1]{Algorithm~\ref*{#1}}}
\newcommand{\thmref}[1]{\hyperref[#1]{Theorem~\ref*{#1}}}
\newcommand{\lemref}[1]{\hyperref[#1]{Lemma~\ref*{#1}}}
\newcommand{\remref}[1]{\hyperref[#1]{Remark~\ref*{#1}}}
\newcommand{\assmref}[1]{\hyperref[#1]{Assumption~\ref*{#1}}}

\declaretheoremstyle[
  spaceabove=6pt,
  spacebelow=6pt,
  headfont=\bfseries,
  bodyfont=\itshape,
  postheadspace=0.5em
]{thmshaded}

\declaretheorem[
  style=thmshaded,
  shaded={rulecolor=black, rulewidth=0.8pt, bgcolor=blue!5!white},
  name=Theorem,
  numberwithin=section
]{theorem}

\declaretheorem[
  style=thmshaded,
  shaded={rulecolor=black, rulewidth=0.8pt, bgcolor=blue!3!white},
  name=Lemma,
  sibling=theorem
]{lemma}

\theoremstyle{definition}

\newtheorem{assumption}[theorem]{Assumption}

\theoremstyle{remark}
\newtheorem{remark}[theorem]{Remark}
\newtheorem{corollary}[theorem]{Corollary}

\usepackage{pifont}

\title{Neural Geometry for PDEs: Regularity, Stability, and Convergence Guarantees}
\author[ISU]{Samundra Karki}
\author[ISU]{\\Adarsh Krishnamurthy}
\author[ISU]{Baskar Ganapathysubramanian\texorpdfstring{\corref{cor1}}{}}
\affiliation[ISU]{organization={Iowa State University}, 
            city={Ames},
            state={Iowa},
            country={USA}}
\cortext[cor1]{Corresponding Authors}
\date{}

\begin{document}
\begin{frontmatter}

% \maketitle

\begin{abstract}
Implicit Neural Representations (INRs) have emerged as a powerful tool for geometric representation, yet their suitability for physics-based simulation remains underexplored. While metrics like Hausdorff distance quantify surface reconstruction quality, they fail to capture the geometric regularity required for provable numerical performance. This work establishes a unified theoretical framework connecting INR training errors to Partial Differential Equation (PDE) (specifically, linear elliptic equation) solution accuracy. We define the minimal geometric regularity required for INRs to support well-posed boundary value problems and derive \emph{a priori} error estimates linking the neural network's function approximation error to the finite element discretization error. Our analysis reveals that to match the convergence rate of linear finite elements, the INR training loss must scale quadratically relative to the mesh size.
\end{abstract}

\end{frontmatter}

\section{Introduction}
\label{sec:introduction}

Accurately representing complex geometry is a foundational requirement in computational science and engineering, underpinning boundary-value problems arising in solid mechanics, fluid flow, transport, and multiphysics coupling. In most classical simulation workflows, geometry is supplied explicitly,for instance, as a CAD boundary representation or a polygonal surface, and subsequently converted into a discretization-ready form, such as a body-fitted mesh. While this pipeline is mature, generating high-quality meshes for complex shapes can be labor-intensive, computationally expensive, and prone to failure or loss of fidelity when geometric complexity increases \citep{nguyen2015isogeometric,liu2022eighty,peskin1972flow,mchenry2008overview, CHIBA1998145}.

Implicit neural representations (INRs) have recently emerged as an attractive alternative for representing geometry in a compact, continuous, and differentiable manner \citep{park2019deepsdf,chibane2020implicit,gropp2020implicit,sitzmann2020implicit}. An INR encodes a shape implicitly as the zero level set of a neural field $\phi_\theta:\mathbb{R}^n \to \mathbb{R}$, which is commonly trained to approximate a signed distance function (SDF). Compared to explicit surface representations, INRs offer several advantages: (i) they provide a resolution-independent description of geometry, (ii) they can be learned directly from data such as point clouds or images \citep{wang2021neus,gropp2020implicit}, and (iii) they can be queried at arbitrary spatial locations without requiring explicit surface extraction.

Implicit fields have long been used in numerical analysis pipelines, most notably for geometry processing and mesh generation. For instance, \citet{persson2005mesh} demonstrated boundary-fitted mesh construction using level-set representations, where the robustness of the procedure depends critically on the well-posedness of geometric quantities derived from the implicit field, such as normals and gradients. More recently, INRs have been integrated directly into embedded and unfitted discretizations. In particular, \citet{karki2025direct,karki2025mechanics} propose direct simulation workflows on octree-based computational grids for prototypical PDEs, including
linear elasticity and the incompressible Navier--Stokes equations, where boundary conditions are imposed weakly through the \emph{Shifted Boundary Method (SBM)} \citep{main2018shifted1,main2018shifted2}. A key ingredient in such direct pipelines is ensuring that the learned implicit field remains geometrically well-behaved in a narrow band around the interface. To this end, \citet{karki2025direct} employs training strategies that enforce an approximately Eikonal structure ($\|\nabla \phi_\theta\|\approx 1$) in a tubular neighborhood, together with additional normal-based consistency losses, in order to maintain stable boundary projections and well-defined normal directions.

These developments highlight that, beyond surface reconstruction, \emph{simulation on neural geometry fundamentally depends on differential regularity properties of the implicit field}. Nevertheless, a rigorous understanding of this requirement is still missing: \emph{what geometric accuracy and regularity are required for an INR to support stable and convergent numerical simulation?} The majority of INR literature evaluates geometric quality using reconstruction metrics such as Chamfer or Hausdorff distance between an original surface and a reconstructed surface. While such measures quantify surface proximity, they do not, by themselves, control the stability constants of boundary value problems or guarantee the well-posedness of boundary operators used in embedded methods. In particular, a small reconstruction error does not preclude pathological behavior near the interface, such as vanishing gradients, irregular normals, or unbounded curvature, all of which can destabilize boundary condition imposition and degrade the accuracy of variational discretizations.

This reveals a gap between \emph{reconstruction fidelity} and \emph{simulation readiness}: there is currently no unified framework that (i) specifies minimal regularity conditions to solve PDE on INR-defined domains, (ii) quantifies how neural approximation error in the geometry propagates into PDE solutions, and (iii) connects training tolerance requirements to discretization accuracy guarantees. The objective of this work is to bridge this gap by developing error analysis that links INR approximation error to PDE solution error for elliptic boundary value problems, thereby clarifying when INRs can serve as reliable geometric proxies for scientific computing. Our analysis focuses on the Poisson-type diffusion model:
\[
-\nabla\cdot(\kappa \nabla u) = f \quad \text{in } \Omega, 
\qquad
u=0 \quad \text{on } \partial\Omega,
\]
and its INR-induced counterpart on $\Omega^p$:
\[
-\nabla\cdot(\kappa \nabla u^p) = f \quad \text{in } \Omega^p, 
\qquad
u^p=0 \quad \text{on } \partial\Omega^p,
\]
where $\kappa$ is uniformly elliptic and $f$ is a given forcing term. Our main contributions are:

\begin{itemize}[leftmargin=2.2em]
\item \textbf{Minimal geometric regularity requirements.}
We identify sufficient geometric conditions on the learned level-set function $\psi$ that ensure $\Gamma^p$ admits a well-defined tubular neighborhood and a unique normal projection map. Concretely, we require $\psi\in C^{1,1}$ in a neighborhood of the interface, a uniform non-degeneracy bound $|\nabla\psi|\ge c_0>0$, and bounded curvature through $\|\nabla^2\psi\|_{\mathrm{op}}\le C_\psi$.

\item \textbf{Training error $\Rightarrow$ Hausdorff geometry bound.}
Let $\phi$ denote a reference (ground-truth) level-set function defining the true boundary $\Gamma=\{\phi=0\}$, that encloses a true domain $\Omega=\{\phi<0\}$. We quantify INR approximation accuracy by the uniform level-set mismatch in a tubular neighborhood $\mathcal{T}_{h} := \{x: |\phi(x)|<h_{tube}\}$:
\[
\varepsilon_\infty := \|\psi-\phi\|_{L^\infty(\mathcal{T}_h)}.
\]
Under gradient non-degeneracy, this implies a Hausdorff perturbation bound of the form
\[
d_H(\Gamma,\Gamma^p)\ \lesssim\ \frac{\varepsilon_\infty}{c_0},
\]
making explicit the amplification mechanism caused by vanishing gradients.

\item \textbf{Main a priori total error estimate (INR + FEM).}
Let $u_h^p$ denote the finite element solution of polynomial degree $k$ on the perturbed domain $\Omega^p$ using mesh size $h_{\mathrm{mesh}}$. Our main result yields the total error decomposition
\begin{equation}
\label{eq:intro_total_error}
\|(u-u_h^p)\|_{L^2(D)}
\ \le\
C_{\mathrm{geom}}\left(\frac{\varepsilon_\infty}{c_0}\right)
\ +\
C_{\mathrm{FEM}}\, h_{\mathrm{mesh}}^{k+1}\, |u^p|_{H^{k+1}(\Omega^p)},
\end{equation}
where $D$ is a fixed background domain containing $\Omega\cup\Omega^p$ and functions are compared via zero extension. A direct consequence is a practical training guideline: to preserve the $\mathcal{O}(h^2_{\mathrm{mesh}})$ FEM convergence rate, the INR must satisfy $\varepsilon_\infty=\mathcal{O}(h_{\mathrm{mesh}}^{2})$ (for linear basis functions).
\end{itemize}

\paragraph{Outline}
The remainder of the paper is organized as follows. In \secref{sec:preliminaries} we introduce notation, geometric distances, and the functional analytic setting. \secref{sec:geometric_regularity} states the minimal geometric regularity assumptions required for well-posed projection and curvature control of INR-defined boundaries. \secref{sec:perturbation} derives a Hausdorff-type bound linking uniform INR function error $\varepsilon_\infty$ to geometric boundary perturbations. \secref{sec:stability} combines sharp domain perturbation estimates for elliptic problems with standard finite element approximation theory to prove the total a priori bound \eqnref{eq:intro_total_error}, yielding explicit requirements on INR training accuracy to avoid degrading simulation convergence.

\section{Preliminaries and Notation}
\label{sec:preliminaries}

We summarize the notation and functional-analytic framework used throughout the paper. This section is primarily added for readers coming from different backgrounds. 

\paragraph{Geometric notation}
Let $\mathbb{R}^n$ denote the $n$-dimensional Euclidean space. Vectors are denoted by bold symbols $\mathbf{x}\in\mathbb{R}^n$. Let $\Omega\subset\mathbb{R}^n$ be an open, bounded domain with boundary $\Gamma := \partial\Omega$. The neural (INR-induced) approximation of $\Omega$ is denoted by $\Omega^p$, with boundary $\Gamma^p$.

For $\mathbf{x}\in\mathbb{R}^n$ and $r>0$, $B_r(\mathbf{x}) := \{\mathbf{y}\in\mathbb{R}^n : |\mathbf{x}-\mathbf{y}|<r\}$ denotes the open ball. The unit sphere is $\mathbb{S}^{n-1} := \{\mathbf{x}\in\mathbb{R}^n : |\mathbf{x}|=1\}$. To quantify geometric discrepancies between domains, we use the Hausdorff distance
\begin{equation}
d_H(\Omega,\Omega^p) := 
\max\Big\{
\sup_{\mathbf{x}\in\Omega}\mathrm{dist}(\mathbf{x},\Omega^p),
\sup_{\mathbf{y}\in\Omega^p}\mathrm{dist}(\mathbf{y},\Omega)
\Big\},
\end{equation}
where $\mathrm{dist}(\mathbf{x},S) := \inf_{\mathbf{y}\in S}|\mathbf{x}-\mathbf{y}|$.

\paragraph{Regularity assumptions}
We use standard H\"older spaces $C^{k,\alpha}$. A function is in $C^{0,1}$ if it is Lipschitz continuous, and in $C^{1,1}$ if its gradient is Lipschitz continuous. Throughout, we assume the level-set functions $\phi$ and $\psi$ satisfy the tubular regularity conditions of \assmref{ass:geom}, which imply that the interfaces $\Gamma=\{\phi=0\}$ and $\Gamma^p=\{\psi=0\}$ are $C^{1,1}$ hypersurfaces. For the Hessian $\nabla^2 \phi$, we denote by $\|\nabla^2 \phi\|_{\mathrm{op}}$ its operator (spectral) norm.

\paragraph{Function spaces}
Let $L^2(\Omega)$ denote the space of square-integrable functions on $\Omega$. We use standard Sobolev spaces $H^k(\Omega)$, with $H^1_0(\Omega)$ denoting functions in $H^1(\Omega)$ that vanish on $\Gamma$ (in the trace sense), and $H^{-1}(\Omega)$ the dual
space of $H^1_0(\Omega)$. In the analysis of elliptic problems on Lipschitz domains, we employ the Besov space $B^{3/2}_{2,\infty}$, characterized (after zero extension to a background domain $D$) by
\begin{equation}
\|u\|_{B^{3/2}_{2,\infty}(D)}
:=
\|u\|_{H^1(D)}
+
\sup_{0<|h|\le1}
\frac{\|\nabla u(\cdot+h)-\nabla u(\cdot)\|_{L^2(D)}}{|h|^{1/2}}.
\end{equation}
This is the sharp regularity scale for elliptic problems posed on Lipschitz domains and plays a central role in the stability estimates of ~\secref{sec:stability}.

\begin{remark}[Zero extension]
Functions defined on $\Omega$ or $\Omega^p$ are extended by zero to a fixed background domain $D\supset\Omega\cup\Omega^p$. This allows solutions on different domains to be compared in a common space $H^1_0(D)$.
\end{remark}

\section{Geometric Regularity of Implicit Neural Representations}
\label{sec:geometric_regularity}
\begin{figure}
    \centering
    \includegraphics[width=0.8\linewidth]{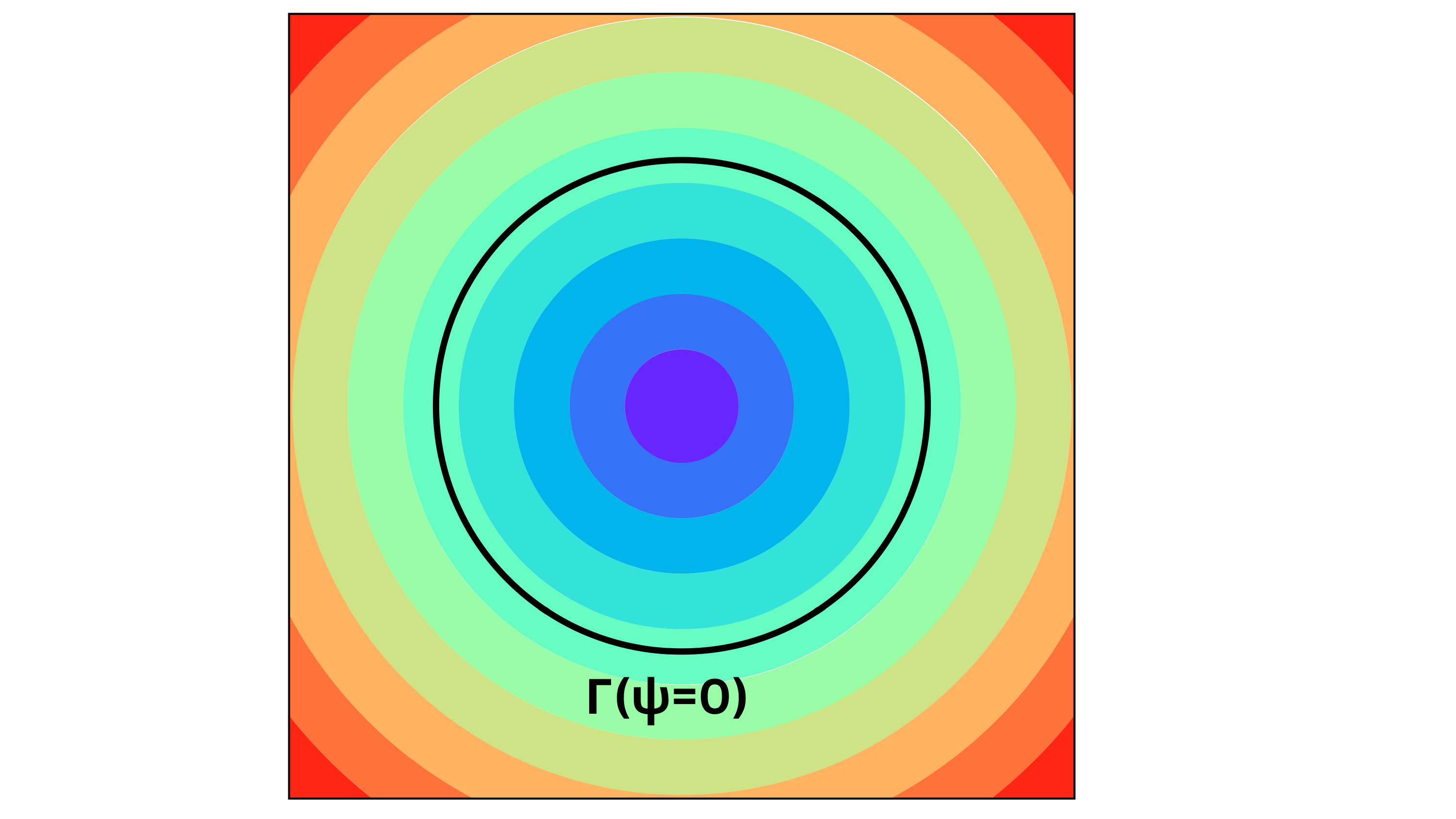}
    \caption{\textbf{Level-Set around true geometry $\Gamma.$} Implicit field ($\psi$) around a boundary ($\Gamma$) such that $\psi=0$ at $\Gamma$ }
    \label{fig:level-set-bdry}
\end{figure}
Let $\psi:\R^n \to \R$ be a scalar field implicitly defining a physical domain $\Omega$ and its interface $\Gamma (\psi =0 )$ as shown in \figref{fig:level-set-bdry}:
\[
\Omega = \{\mathbf{x}\in\R^n:\psi(\mathbf{x})<0\}, 
\qquad 
\Gamma = \{\mathbf{x}\in\R^n:\psi(\mathbf{x})=0\}.
\]
% For the level-set $\psi$ to serve as a valid representation for a PDE problem involving Neumann or Dirichlet conditions, it must be a smooth hypersurface with a well-defined normal vector $\mathbf{n} = \nabla \psi / \|\nabla \psi\|$ at $\Gamma$. We formally define the necessary conditions required for $\psi$.
For the level-set field $\psi$ to serve as a valid geometric representation for PDE problems with Dirichlet or Neumann boundary conditions imposed on $\Gamma$, it must satisfy certain regularity and non-degeneracy properties in a neighborhood of the boundary. We now formally state the assumptions required of $\psi$.
\begin{assumption}[Geometric Regularity]\label{ass:geom}
Let $\mathcal{T}_h = \{\mathbf{x} \in \R^n : |\psi(\mathbf{x})|<h\}$ be a tubular neighborhood of width $h>0$ around $\Gamma$ (which is zero-level set of a scalar-field  $\psi$). We assume $\psi \in C^{1,1}(\mathcal{T}_h)$, meaning $\psi$ is continuously differentiable with a Lipschitz continuous gradient. Furthermore, there exist constants $c_0>0$ and $C_\psi<\infty$ such that for all $\mathbf{x} \in \mathcal{T}_h$:
\begin{equation}\label{eq:geom_assump}
\|\nabla \psi(\mathbf{x})\| \ge c_0,
\qquad 
\|\nabla^2 \psi(\mathbf{x})\|_{\op} \le C_\psi,
\end{equation}
where $\|\cdot\|_{\op}$ denotes the spectral norm of the Hessian.
\end{assumption}

These conditions imply that the unit normal field $\mathbf{n}$ is Lipschitz continuous and the principal curvatures of $\Gamma$ are bounded by $C_\psi/c_0$.

\begin{remark}[Valid Neural Architectures]
The choice of activation function determines whether an INR satisfies \assmref{ass:geom}:
\begin{enumerate}
    \item \textbf{Smooth Activations ($C^\infty$):} Architectures using \texttt{tanh}, \texttt{sine} (SIREN), \texttt{swish}, or \texttt{GELU} yield $\psi \in C^\infty$, satisfying the assumption locally provided gradients do not vanish.
    \item \textbf{Higher-Order Rectifiers ($C^{1,1}$):} The Squared-ReLU, $\sigma(x) = (\max\{0, x\})^2$, yields globally $C^{1,1}$ functions with bounded Hessians, making them geometrically safe for PDE solvers.
    \item \textbf{Standard ReLU ($C^{0,1}$):} Standard \texttt{ReLU} networks produce piecewise linear boundaries with sharp creases. They fail \assmref{ass:geom} as curvature is distributional (Dirac delta) and normals are discontinuous, leading to ill-posed boundary projections.
\end{enumerate}
\end{remark}

\subsection{Well-Posedness of the Boundary Projection}
\label{sec:projection}

Fictional-domain methods (e.g., Shifted Boundary Method) require mapping points near the boundary back to $\Gamma$ \citep{main2018shifted1,main2018shifted2,hsu2016direct,JAISWAL2024117426}. We show that \assmref{ass:geom} guarantees this mapping is unique, even if $\psi$ is not a perfect Signed Distance Function (SDF).

\begin{lemma}[Local existence and uniqueness of the normal projection]
\label{lem:projection}
Let $\psi \in C^{1,1}(\mathcal{T}_h)$ satisfy \assmref{ass:geom} whose zero-level set represents boundary of domain $\Omega$. Then, there exists a constant $h_{\mathrm{crit}} > 0$, depending only on $c_0$ and $C_\psi$, such that for every $h \in (0,h_{\mathrm{crit}})$ the following holds:

For every $\mathbf{x}\in \mathcal{T}_h$, there exist unique $\mathbf{x}^* \in \Gamma$
and unique $d \in (-h,h)$ such that
\begin{equation}
\label{eq:projection_relation}
\mathbf{x} = \mathbf{x}^* + d\,\mathbf{n}(\mathbf{x}^*),
\qquad
\mathbf{n}(\mathbf{x}^*) = \frac{\nabla \psi(\mathbf{x}^*)}{\|\nabla \psi(\mathbf{x}^*)\|}.
\end{equation}
% Moreover, the projection map $\Pi_\Gamma:\mathcal{T}_h \to \Gamma$, $\Pi_\Gamma(\mathbf{x})=\mathbf{x}^*$, is Lipschitz on $\mathcal{T}_h$.
\end{lemma}

The full proof is provided in \ref{app:projection}.
% \begin{proof}[Proof sketch]
% Under Assumption~\ref{ass:geom}, the zero level set $\Gamma=\{\psi=0\}$ is a $C^{1,1}$ embedded hypersurface with Lipschitz continuous unit normal and uniformly bounded curvature. Classical tubular-neighborhood (positive-reach) results for $C^{1,1}$ hypersurfaces imply the existence of a radius $h_{\mathrm{crit}}>0$, depending only on $c_0$ and $C_\psi$, such that the nearest-point projection onto $\Gamma$ is single-valued in $\mathcal T_h$ for all $h<h_{\mathrm{crit}}$. Within this neighborhood, every point admits a unique normal-coordinate representation of the form \eqnref{eq:projection_relation}. A complete proof is provided in~\ref{app:projection}.
% \end{proof}

This result confirms that INRs need not strictly satisfy the Eikonal equation $\|\nabla \psi\|=1$. As long as gradients are non-vanishing ($c_0 > 0$) and curvature is controlled, the geometry is valid for simulation. 

In \secref{sec:error_bounds}, we use INR based geometric representation with unfitted methods, mainly \textbf{SBM}\citep{main2018shifted1,main2018shifted2} with octrees as background mesh. The method precisely computes $\mathbf{x}-\mathbf{x^*}$ to enforce boundary condition.

\section{Geometric Perturbation Analysis}
\label{sec:perturbation}

In practice, we approximate the true geometry $\phi$ with a neural network $\psi$. We now quantify the geometric error introduced by this approximation. Let:
\[
\phi(\mathbf{x}) \quad \text{(True Level Set)}, \qquad \psi(\mathbf{x}) \quad \text{(Neural Approximation)}.
\]
We assume both fields satisfy \assmref{ass:geom} in a common tube $\mathcal{T}_h$. Furthermore, we assume the network is trained to a uniform tolerance:
\begin{equation}\label{eq:compatibility}
\|\psi - \phi\|_{L^\infty(\mathcal{T}_h)} \le \varepsilon_\infty.
\end{equation}

\begin{lemma}[Hausdorff Distance Bound]\label{lem:hausdorff}
Let $\Gamma = \{\phi=0\}$ represent true zero-level set and $\Gamma^p = \{\psi=0\}$ represent corresponding zero-level set given by a neural approximation. Under the regularity conditions of \assmref{ass:geom} and the compatibility \eqnref{eq:compatibility}, the Hausdorff distance is bounded by:
\begin{equation}
\distH(\Gamma, \Gamma^p) \le \frac{\varepsilon_\infty}{\min(c_0, \tilde{c}_0)} + \mathcal{O}(\varepsilon_\infty^2),
\end{equation}
where $c_0, \tilde{c}_0$ are the lower bounds on $\|\nabla \phi\|$ and $\|\nabla \psi\|$ respectively.
\end{lemma}

The full proof is provided in \ref{app:hausdorff}.

\begin{remark}[Impact of Vanishing Gradients]
The bound is scaled by $1/c_0$. This highlights a critical training requirement: merely minimizing the residual $|\psi - \phi|$ is insufficient. The training loss must also enforce the Eikonal constraint or gradient regularization to prevent $c_0 \to 0$. If the INR gradients vanish at the interface, the geometric error $\distH$ becomes unbounded, destroying simulation accuracy.
\end{remark}

\section{Stability and Discretization on Perturbed Domains}
\label{sec:stability}

In this section we quantify how geometric approximation errors in an INR-defined domain
propagate to the solution of elliptic boundary value problems. Our presentation follows the
quantitative domain perturbation framework of Savaré and Schimperna~\citep{savare2002domain},
which yields sharp stability estimates under minimal geometric regularity (Lipschitz boundaries).

\subsection{Motivation}
Let $u$ denote the weak solution of an elliptic problem posed on a domain $\Omega\subset\R^n$.
A fundamental question in stability theory and numerical analysis is:

\begin{quote}
How does the solution $u$ change when $\Omega$ is replaced by a perturbed domain $\Omega^p$?
\end{quote}

This question arises in many classical settings, including stability of variational solutions
on varying domains, shape optimization~\citep{henrot1994continuity}, and error analysis of
polygonal boundary approximation~\citep{ciarlet1972interpolation}. In our setting, the
perturbation $\Omega^p$ comes from an INR-reconstructed geometry. While qualitative tools
such as $\Gamma$-convergence or Mosco convergence yield convergence statements, our goal is a
\emph{quantitative} bound that depends explicitly on:
(i) a distance between domains, (ii) boundary regularity parameters, and (iii) norms of the forcing $f$.

\subsection{Problem formulation}
Let $D\subset\R^n$ be a bounded background domain containing both the true domain $\Omega$
and the perturbed domain $\Omega^p$. We consider the Dirichlet problems
\begin{align}
\label{eq:pde_true_dir_rewrite}
\text{(True)}\quad
\begin{cases}
-\nabla\cdot(\kappa\nabla u)=f & \text{in }\Omega,\\
u=0 & \text{on }\Gamma:=\partial\Omega,
\end{cases}
\qquad
\text{(Perturbed)}\quad
\begin{cases}
-\nabla\cdot(\kappa\nabla u^p)=f & \text{in }\Omega^p,\\
u^p=0 & \text{on }\Gamma^p:=\partial\Omega^p,
\end{cases}
\end{align}
where $\kappa\in L^\infty(D)$ satisfies uniform ellipticity
\[
0<\kappa_{\min}\le \kappa(x)\le \kappa_{\max}<\infty,
\]
and $f\in L^2(D)$. We extend solutions by zero outside their domains so that
$u,u^p\in H^1_0(D)$ for comparison on $D$.

\begin{remark}[Level-set convention]
In our INR setting we represent the domain by a level set:
$\Omega=\{\phi<0\}$ and $\Omega^p=\{\psi<0\}$, so that
$\Gamma=\partial\Omega=\{\phi=0\}$ and $\Gamma^p=\partial\Omega^p=\{\psi=0\}$
(in the absence of topological pathologies).
\end{remark}

\begin{remark}[Non-homogeneous boundary conditions]
For boundary data $u|_{\Gamma}=g$, introduce a lifting $w\in H^1(D)$ with $w|_{\Gamma}=g$
(and $w|_{\Gamma^p}\approx g$), and apply the analysis to the homogenized unknowns
$\tilde u=u-w$ and $\tilde u^p=u^p-w$.
\end{remark}

\subsection{Geometric regularity: uniform cone condition}
The stability theory of~\citet{savare2002domain} requires a minimal boundary regularity
assumption formulated via a uniform cone condition.

\begin{assumption}[Uniform $(\rho,\theta)$-cone condition]
\label{ass:cone}
The domain $\Omega\subset\mathbb{R}^n$ satisfies a uniform $(\rho,\theta)$-cone condition if
there exist constants $\rho>0$ and $\theta\in(0,\pi/2)$ such that for every $x_0\in\partial\Omega$
there exists a unit vector $n(x_0)\in\mathbb S^{n-1}$ such that, for all
\[
h \in C_{\rho,\theta}(n(x_0))
:= \big\{ h\in\mathbb{R}^n : 0<|h|<\rho,\; h\cdot n(x_0) > |h|\cos\theta \big\},
\]
the inclusions
\begin{equation}
\label{eq:cone_inclusions}
(B_{3\rho}(x_0)\cap \Omega) - h \subset \Omega,
\qquad
(B_{3\rho}(x_0)\setminus \Omega) + h \subset \mathbb{R}^n\setminus \Omega
\end{equation}
hold.
\end{assumption}

The cone condition is equivalent to uniform Lipschitz regularity of $\partial\Omega$; see
\cite[Definition~2.6 and Remark~2.7]{savare2002domain}.

\begin{remark}[Relation to INR regularity]
Our level-set regularity assumption (\assmref{ass:geom}) implies that $\partial\Omega$ is a
$C^{1,1}$ hypersurface, hence Lipschitz, and therefore satisfies \assmref{ass:cone}; see,
e.g.,~\cite[Theorem~I.2.2.2]{grisvard1985elliptic}.
\end{remark}

\subsection{Main stability estimate}
Let
\[
d_\star := d_H(\Gamma,\Gamma^p)
\]
denote the Hausdorff distance between boundaries.

\begin{theorem}[Stability under domain perturbation {\cite[Theorems~1--2 and Corollary~8.2]{savare2002domain}}]
\label{thm:pde_stability}
Assume that $\Omega\subset\mathbb{R}^n$ satisfies \assmref{ass:cone}.
Suppose that the perturbation is sufficiently small:
\begin{equation}
\label{eq:small_perturb}
d_\star \le \tfrac12\,\rho\sin\theta.
\end{equation}
Let $u\in H^1_0(\Omega)$ and $u^p\in H^1_0(\Omega^p)$ be the weak solutions of
\eqnref{eq:pde_true_dir_rewrite}. Then, after extension by zero to $D$, the following
estimates hold:
\begin{align}
\label{eq:energy_bound_clean}
\|\nabla u - \nabla u^p\|_{L^2(D)}
&\le C_1\,
\|f\|_{L^2(D)}^{1/2}\|f\|_{H^{-1}(D)}^{1/2}
\left(\frac{d_\star}{\rho\sin\theta}\right)^{1/2},
\\
\label{eq:L2_bound_clean}
\|u - u^p\|_{L^2(D)}
&\le C_2\,
\|f\|_{L^2(D)}^{1/2}\|f\|_{H^{-1}(D)}^{1/2}
\left(\frac{d_\star}{\rho\sin\theta}\right),
\end{align}
where $C_1$ and $C_2$ depend only on $n$, $\kappa_{\max}/\kappa_{\min}$, and $\operatorname{diam}(D)$.
\end{theorem}

We briefly summarize the mechanism underlying \thmref{thm:pde_stability}; see
\citet{savare2002domain} for full proofs.

\paragraph{Variational formulation in a common space.}
Let $V:=H^1_0(D)$ and define
\[
a(v,w):=\int_D \kappa \nabla v\cdot \nabla w\,dx,
\]
which is coercive and continuous on $V$.
Although $u$ and $u^p$ solve problems on different domains, we may view them as posed in
different closed subspaces
\[
V_\Omega := H^1_0(\Omega)\subset V,
\qquad
V_{\Omega^p}:=H^1_0(\Omega^p)\subset V.
\]
An abstract approximation argument yields a bound in terms of distances to these subspaces.

\begin{lemma}[Abstract stability {\cite[Lemma~4.2]{savare2002domain}}]
\label{lem:abstract_stability}
There exists $\sigma>0$ depending only on the coercivity/continuity constants of $a(\cdot,\cdot)$ such that
\[
\|u-u^p\|_{V}
\le \sigma\Bigl(d(u,V_{\Omega^p}) + d(u^p,V_\Omega)\Bigr),
\qquad
d(v,W):=\inf_{w\in W}\|v-w\|_V.
\]
\end{lemma}

\paragraph{Excess, complementary excess, and Hausdorff control.}
Directed set discrepancies are quantified via the excess
\[
e(X,Y):=\sup_{x\in X}\dist(x,Y),
\qquad
\dist(x,Y):=\inf_{y\in Y}|x-y|,
\]
and the complementary excess $\check e(X,Y):=e(\R^n\setminus Y,\R^n\setminus X)$.
Under \assmref{ass:cone} and \eqnref{eq:small_perturb}, the complementary excess between
$\Omega$ and $\Omega^p$ is controlled by the Hausdorff distance; see \cite[Lemma~7.4]{savare2002domain}:
\[
\check e(\Omega,\Omega^p)\ \lesssim\ \frac{d_\star}{\sin\theta}.
\]

\paragraph{Localization and shift construction.}
Using a partition of unity and the cone condition, one constructs comparison functions by
shifting locally in admissible cone directions, producing difference-quotient terms with
$|h|\sim d_\star/\sin\theta$; see \cite[Proposition~5.1 and Lemma~7.3]{savare2002domain}.

\paragraph{Difference quotient regularity in Lipschitz domains.}
Elliptic solutions on Lipschitz domains satisfy optimal difference quotient estimates
corresponding to Besov regularity $B^{3/2}_{2,\infty}$; see
\cite[Theorem~6.2]{savare2002domain} and \citet{jerison1995}.
In particular,
\[
\|u(\cdot+h)-u(\cdot)\|_{L^2}+\|\nabla u(\cdot+h)-\nabla u(\cdot)\|_{L^2}
\ \lesssim\ |h|^{1/2}\,\|f\|_{L^2}.
\]

\paragraph{Assembling the estimate; duality for the \(L^2\) bound.}
Combining Lemma~\ref{lem:abstract_stability} with the localization and difference quotient
bounds yields the energy estimate \eqnref{eq:energy_bound_clean}.
The $L^2$ estimate \eqnref{eq:L2_bound_clean} follows by a duality (Aubin--Nitsche type)
argument adapted to perturbed domains; see \cite[Theorem~8.4]{savare2002domain}.

\begin{remark}[Optimality of the $1/2$ exponent]
The square-root dependence $d_\star^{1/2}$ in \eqnref{eq:energy_bound_clean} is sharp in general,
reflecting the regularity barrier for elliptic problems on Lipschitz domains; see
\citet{savare2002domain,jerison1995}.
\end{remark}

\subsection{Connection to INR geometry}
\thmref{thm:pde_stability} bridges INR approximation quality and PDE solution accuracy.
Combining it with the geometric perturbation estimate from \secref{sec:perturbation} yields
an explicit INR-to-PDE propagation bound.

\begin{corollary}[INR-to-PDE error propagation in \(L^2\)]
\label{cor:inr_pde_error_L2}
Under Assumptions~\ref{ass:geom} and~\ref{ass:cone}, let $\Omega=\{\phi<0\}$ and
$\Omega^p=\{\psi<0\}$. Assume the INR field $\psi$ satisfies, on a tubular neighborhood
$\mathcal T_h$,
\[
\|\phi-\psi\|_{L^\infty(\mathcal T_h)}\le \varepsilon_\infty.
\]
Then the corresponding solutions $u$ on $\Omega$ and $u^p$ on $\Omega^p$ satisfy
\begin{equation}
\label{eq:inr_to_pde_L2}
\|u-u^p\|_{L^2(D)}
\le
\tilde C_{L^2}\,
\|f\|_{L^2(D)}^{1/2}\,
\|f\|_{H^{-1}(D)}^{1/2}\,
\left(\frac{\varepsilon_\infty}{c_0\,\rho\sin\theta}\right),
\end{equation}
where $c_0:=\inf_{x\in\Gamma}\|\nabla\phi(x)\|$.
\end{corollary}

\begin{proof}
By \lemref{lem:hausdorff}, the uniform mismatch implies
$d_H(\Gamma,\Gamma^p)\lesssim \varepsilon_\infty/c_0$ (up to higher-order terms).
Substituting this bound into Theorem~\ref{thm:pde_stability} in its \(L^2\)-stability form
\eqnref{eq:L2_bound_clean} proves \eqnref{eq:inr_to_pde_L2}.
\end{proof}

\subsection{Total error: geometry + discretization}
In practice we solve the perturbed-domain problem numerically on $\Omega^p$ using a finite
element discretization with mesh size $h_{\mathrm{mesh}}$. Let $u_h^p$ denote the discrete
solution. The total error measured on $D$ admits the decomposition
\[
\|u-u_h^p\|_{L^2(D)}
\le
\|u-u^p\|_{L^2(D)}+\|u^p-u_h^p\|_{L^2(\Omega^p)}.
\]

\begin{theorem}[A priori total \(L^2\) error estimate]
\label{thm:total_error_L2}
Under the assumptions of Corollary~\ref{cor:inr_pde_error_L2}, assume conforming finite
elements of polynomial degree $k$ are used on $\Omega^p$ and that $u^p\in H^{k+1}(\Omega^p)$.
Then there exist constants $C_{\mathrm{geom}},C_{\mathrm{FEM}}>0$, independent of
$h_{\mathrm{mesh}}$ and $\varepsilon_\infty$, such that
\begin{equation}
\label{eq:total_error_bound_L2}
\|u-u_h^p\|_{L^2(D)}
\le
C_{\mathrm{geom}}\,
\|f\|_{L^2(D)}^{1/2}\,
\|f\|_{H^{-1}(D)}^{1/2}\,
\left(\frac{\varepsilon_\infty}{c_0}\right)
+
C_{\mathrm{FEM}}\,h_{\mathrm{mesh}}^{k+1}\,|u^p|_{H^{k+1}(\Omega^p)}.
\end{equation}
\end{theorem}

\begin{proof}
The triangle inequality gives
\[
\|u-u_h^p\|_{L^2(D)}
\le
\|u-u^p\|_{L^2(D)}+\|u^p-u_h^p\|_{L^2(\Omega^p)}.
\]
The first term is bounded by Corollary~\ref{cor:inr_pde_error_L2}.
The second term follows from the standard $L^2$ finite element estimate:
\[
\|u^p-u_h^p\|_{L^2(\Omega^p)}\le
C\,h_{\mathrm{mesh}}^{k+1}\,|u^p|_{H^{k+1}(\Omega^p)}.
\]
Combining the two bounds yields \eqnref{eq:total_error_bound_L2}.
\end{proof}

\begin{remark}[Balancing geometry and discretization errors]
Assume the forcing norms $\|f\|_{L^2(D)}$, $\|f\|_{H^{-1}(D)}$ and the regularity seminorm
$|u^p|_{H^{k+1}(\Omega^p)}$ are $\mathcal O(1)$, and that the constants
$C_{\mathrm{geom}}$ and $C_{\mathrm{FEM}}$ are comparable. Balancing the two terms in
\eqnref{eq:total_error_bound_L2} yields the guideline
\[
\varepsilon_\infty \sim h_{\mathrm{mesh}}^{k+1}.
\]
In particular, for $k=1$ one expects $\varepsilon_\infty\sim h_{\mathrm{mesh}}^2$, while for
$k=2$ one expects $\varepsilon_\infty\sim h_{\mathrm{mesh}}^3$.
\end{remark}

\begin{remark}[Practical implications (order-of-magnitude)]
The bound \eqnref{eq:total_error_bound_L2} shows that INR geometry error dominates unless
$\varepsilon_\infty$ is reduced in tandem with the mesh resolution. For example, if one aims
for an $L^2$ discretization error of order $h_{\mathrm{mesh}}^{k+1}$, then $\varepsilon_\infty$
should be chosen of the same order (under the assumptions of the previous remark). Thus:
\begin{itemize}
\item For $k=1$ and $h_{\mathrm{mesh}}=10^{-2}$, one targets $\varepsilon_\infty\approx 10^{-4}$.
\item For $k=2$ and $h_{\mathrm{mesh}}=10^{-2}$, one targets $\varepsilon_\infty\approx 10^{-6}$.
\item For $k=2$ and $h_{\mathrm{mesh}}=10^{-3}$, one targets $\varepsilon_\infty\approx 10^{-9}$.
\end{itemize}
These scalings motivate training strategies that adapt INR tolerance to the intended
discretization accuracy, avoiding over-training when coarse meshes are sufficient.
\end{remark}

\section{Result}
\label{sec:results}

This section provides numerical evidence supporting the geometric assumptions and
perturbation bounds developed in \secref{sec:perturbation}.
We proceed in two stages.
First, using analytic implicit fields in 2D, we isolate the role of geometric regularity
(\assmref{ass:geom}) in ensuring that the normal/closest-point projection is
well-defined and stable (\lemref{lem:projection}); we visualize where the gradient
non-degeneracy holds and verify projection behavior under Newton iteration.
Second, we repeat the same diagnostics for SIREN-based INRs trained to approximate a
unit-circle signed distance field, and empirically validate the Hausdorff perturbation
estimate (\lemref{lem:hausdorff}) across increasing training budgets.
Finally, we connect geometric mismatch to downstream PDE accuracy by solving a Poisson
problem on perturbed geometries with an unfitted octree discretization and shifted
boundary method (SBM) boundary treatment, demonstrating how solution error scales with
Hausdorff distance and how mesh refinement saturates when geometry error dominates.
\subsection{Numerical validation of closest projection}
% \begin{figure*}[t]
% \centering
% \setlength{\tabcolsep}{3pt}
% \renewcommand{\arraystretch}{1.0}

\begin{figure*}[t]
\centering
\setlength{\tabcolsep}{4pt}
\renewcommand{\arraystretch}{1.0}

\begin{tabular}{ccc}
\includegraphics[width=0.31\textwidth]{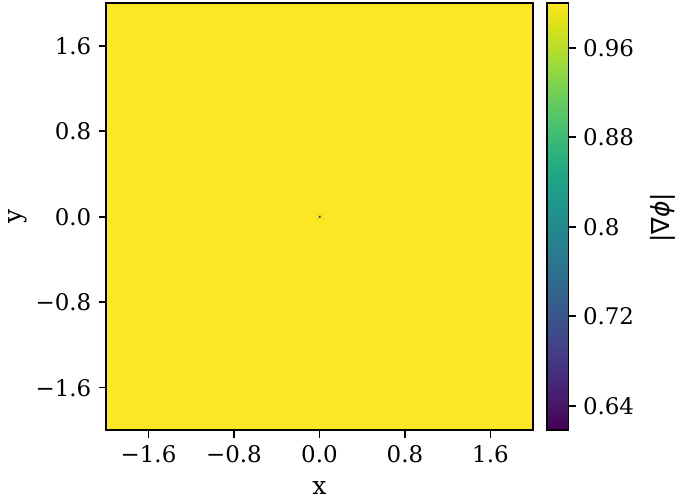} &
\includegraphics[width=0.31\textwidth]{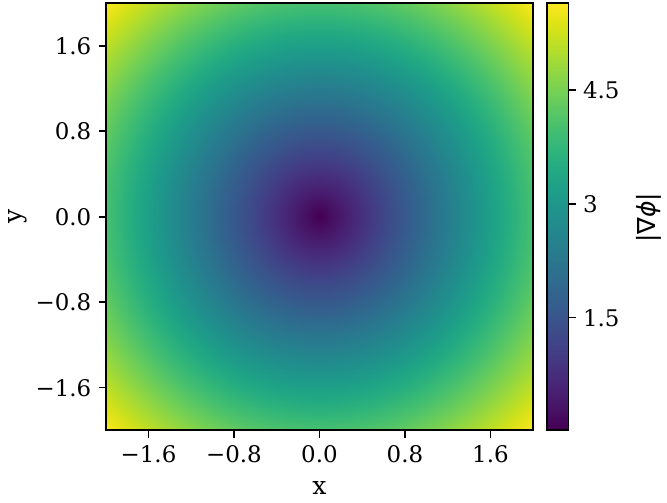} &
\includegraphics[width=0.31\textwidth]{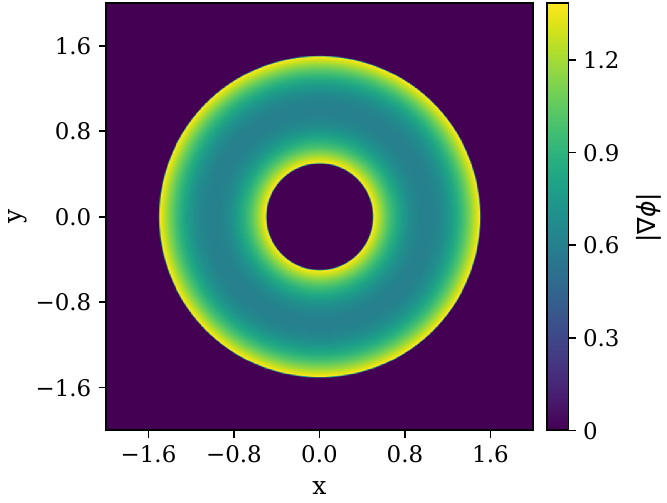} \\[-2pt]
\includegraphics[width=0.31\textwidth]{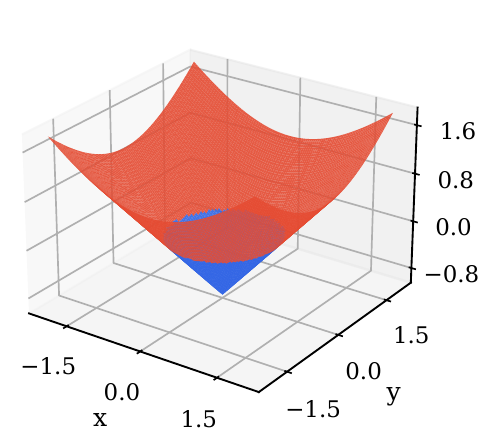} &
\includegraphics[width=0.31\textwidth]{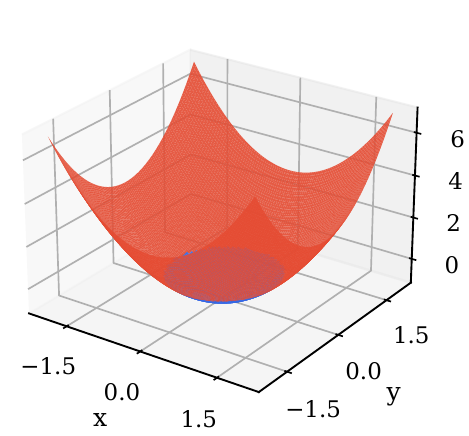} &
\includegraphics[width=0.31\textwidth]{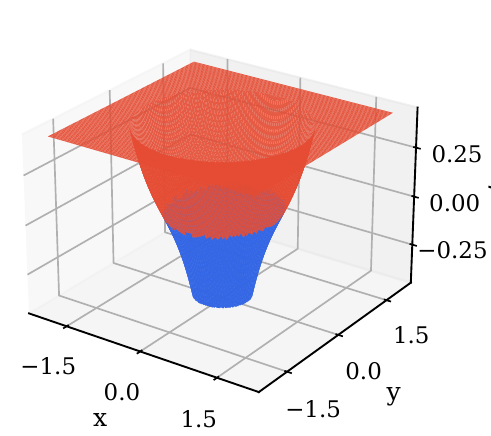} \\
\end{tabular}

\vspace{2pt}
\caption{\textbf{Validation of \assmref{ass:geom} on three analytic level-set constructions.}
Columns: (Case 1) Eikonal signed-distance field $\phi=d$,
(Case 2) non-eikonal quadratic $\phi=r^2-R^2$,
(Case 3) banded non-degenerate field (designed to maintain $|\nabla\phi|\ge c_0$ only within a prescribed tubular band).
Top row: gradient magnitude $|\nabla\phi|$, illustrating (non-)degeneracy near the zero level set.
Bottom row: surface plot $z=\phi(x,y)$, highlighting the qualitative shape of the scalar field around the interface.
In particular, Case~3 demonstrates that regularity requirements such as the lower bound $|\nabla\phi|\ge c_0$ may hold only in a narrow neighborhood $\mathcal{T}_h$ of the interface.}
\label{fig:assumption_geom_validation}
\end{figure*}

\begin{figure*}[t]
\centering
\setlength{\tabcolsep}{3pt}
\renewcommand{\arraystretch}{1.0}

% ---------------- Row 1: success masks ----------------
\begin{tabular}{ccc}
\includegraphics[width=0.28\textwidth]{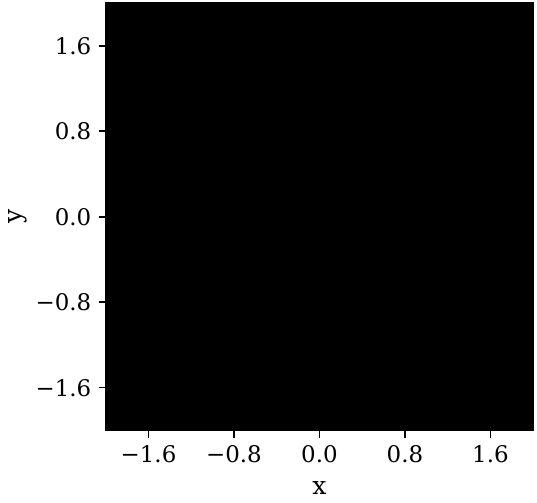} &
\includegraphics[width=0.28\textwidth]{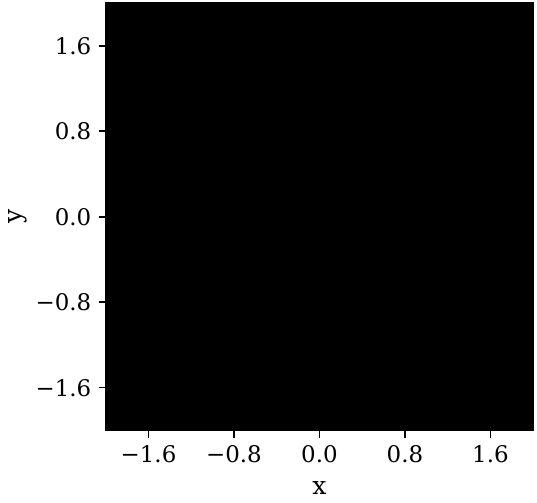} &
\includegraphics[width=0.28\textwidth]{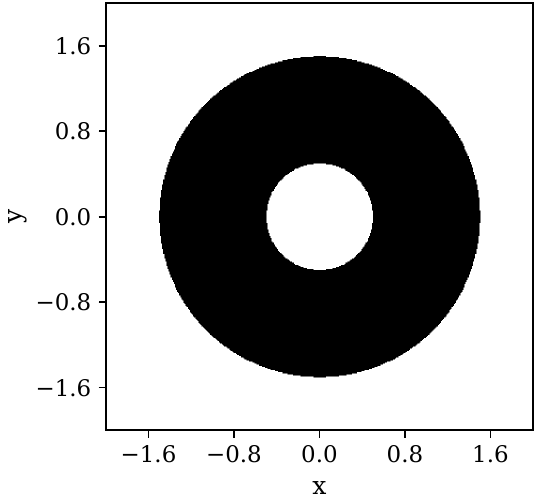} \\
\end{tabular}

\vspace{0.15em}
\includegraphics[width=0.3\textwidth]{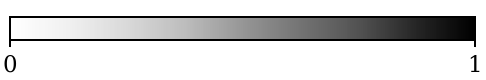}

\vspace{0.35em}

% ---------------- Row 2: error maps ----------------
\begin{tabular}{ccc}
\includegraphics[width=0.28\textwidth]{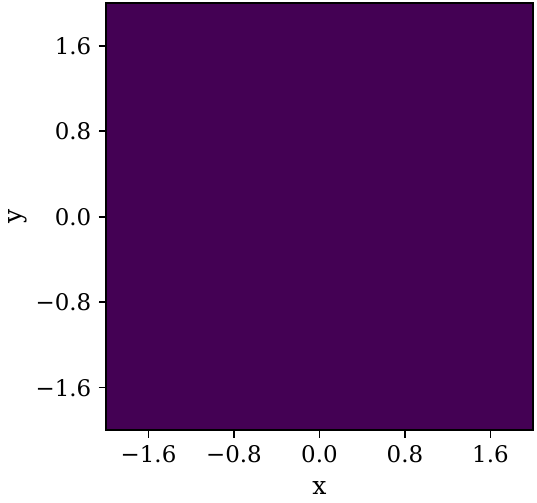} &
\includegraphics[width=0.28\textwidth]{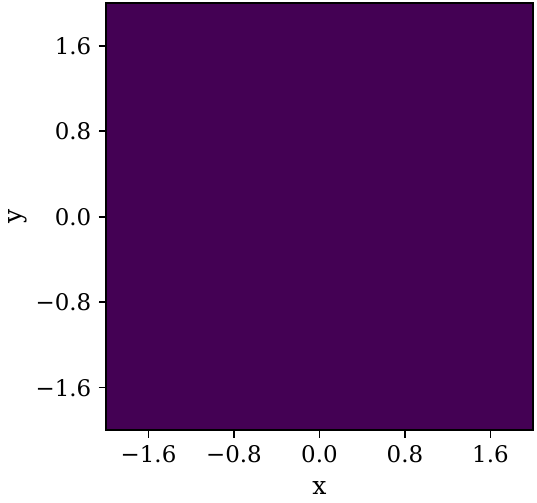} &
\includegraphics[width=0.28\textwidth]{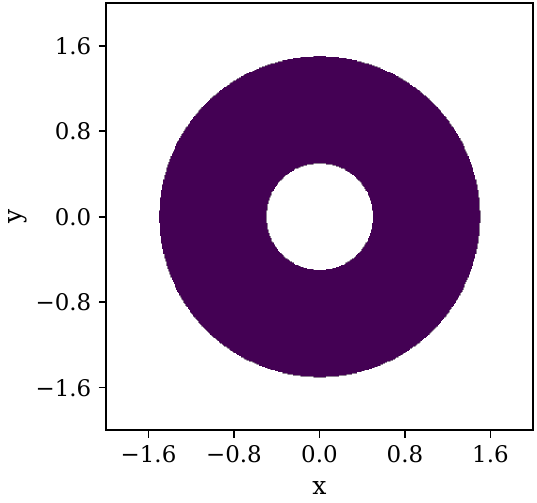} \\
\end{tabular}

\vspace{0.15em}
\includegraphics[width=0.3\textwidth]{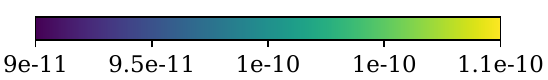}

\vspace{-1mm}
\caption{\textbf{Numerical validation of \lemref{lem:projection} via Newton normal projection.}
Top: success masks (convergence to the true closest point).
Bottom: projection error $\|\mathbf{x}_{\mathrm{final}}-\mathbf{x}^*(\mathbf{x})\|$ on the successful set.
Case~1 (Eikonal SDF) and Case~2 (non-eikonal quadratic) converge robustly in a wide neighborhood of $\Gamma$,
while Case~3 succeeds only inside the prescribed tubular neighborhood $\mathcal{T}_h$.}
\label{fig:projection_validation}
\end{figure*}

Before turning to INR-defined geometries, we first illustrate the role of
\assmref{ass:geom} and \lemref{lem:projection} on three analytic level-set
constructions in two dimensions.
\figref{fig:assumption_geom_validation} compares: (i) an Eikonal signed distance field
$\phi=d$, (ii) a non-Eikonal quadratic implicit field $\phi=r^2-R^2$ sharing the same zero
level set, and (iii) a \emph{banded} level-set field constructed to satisfy
$|\nabla\phi|\ge c_0$ only in a prescribed tubular neighborhood $\mathcal{T}_h$ around the
interface.
The gradient magnitude plots explicitly highlight where the non-degeneracy condition
$|\nabla\phi|\ge c_0$ holds and where it fails, while the surface visualizations
$z=\phi(x,y)$ provide an intuitive view of the scalar-field landscape near the zero level
set.

We next validate the well-posedness of the normal projection map numerically using the
standard Newton closest-point iteration,
\[
\mathbf{x}_{k+1}
=
\mathbf{x}_k
-
\frac{\phi(\mathbf{x}_k)}{\|\nabla\phi(\mathbf{x}_k)\|^2}\,\nabla\phi(\mathbf{x}_k),
\]
which coincides with Newton’s method for solving $\phi(\mathbf{x})=0$ along the normal
direction.
~\figref{fig:projection_validation} reports (top) the set of initial points for which
the iteration converges to the \emph{true} closest point on $\Gamma=\{\phi=0\}$ and
(bottom) the corresponding projection error.
For Case~1 and Case~2, the iteration converges robustly in a wide neighborhood of the
interface, consistent with $\nabla\phi$ being non-vanishing near $\Gamma$.
In contrast, Case~3 exhibits reliable convergence only within the designed band
$\mathcal{T}_h$, illustrating that projection-based constructions are stable only when the
analysis is restricted to $h<h_{\mathrm{crit}}$ and the gradient non-degeneracy condition
in \assmref{ass:geom} is enforced.

\subsection{Numerical validation of the Hausdorff perturbation bound}
\begin{figure*}[t]
\centering
\setlength{\tabcolsep}{3pt}
\renewcommand{\arraystretch}{1.0}

% -------- Row 1: |\nabla \psi_\theta| --------
\begin{tabular}{cccc}
\includegraphics[width=0.235\textwidth]{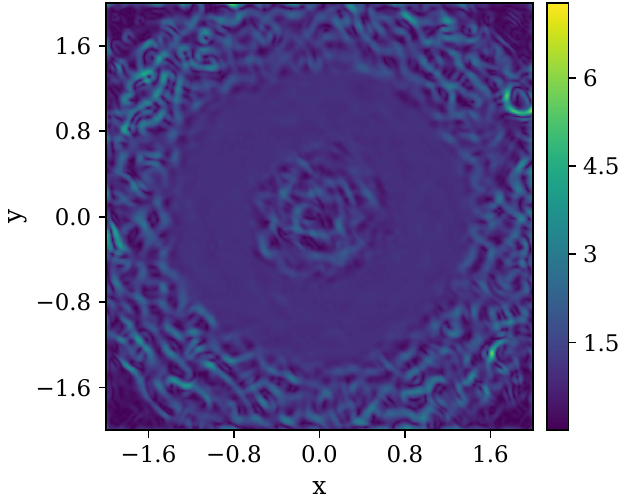} &
\includegraphics[width=0.235\textwidth]{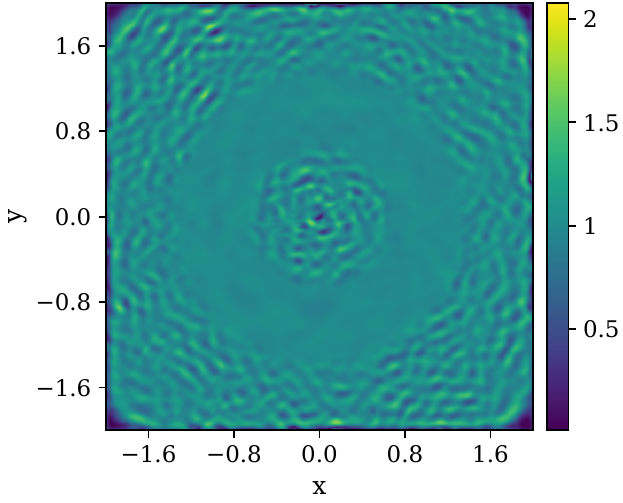} &
\includegraphics[width=0.235\textwidth]{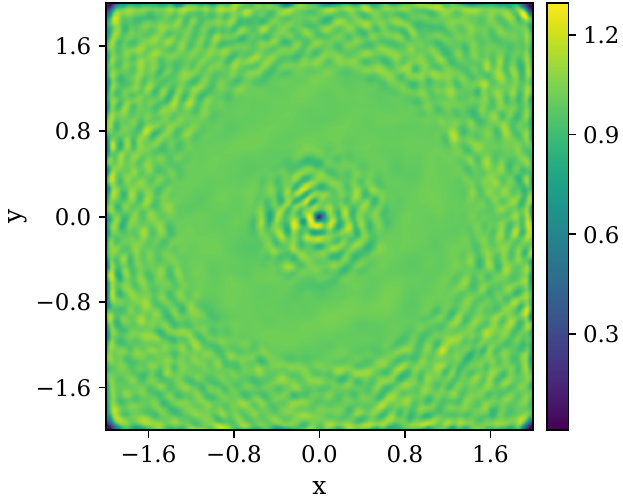} &
\includegraphics[width=0.235\textwidth]{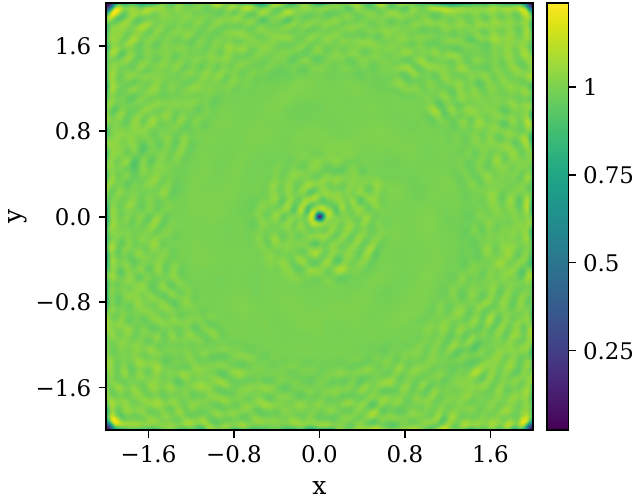} \\
\end{tabular}

\vspace{0.25em}

% -------- Row 2: 3D surface z = psi_theta(x,y) --------
\begin{tabular}{cccc}
\includegraphics[width=0.235\textwidth]{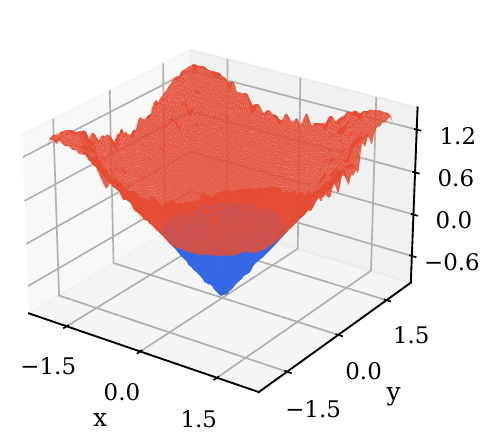} &
\includegraphics[width=0.235\textwidth]{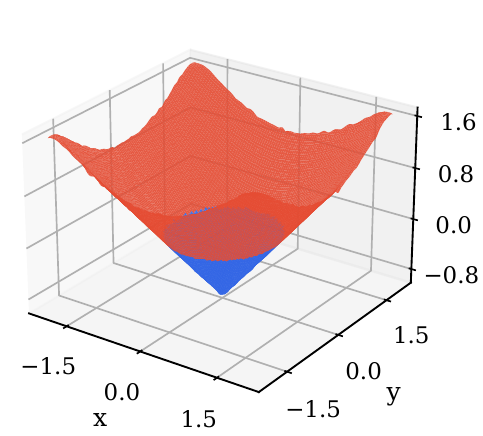} &
\includegraphics[width=0.235\textwidth]{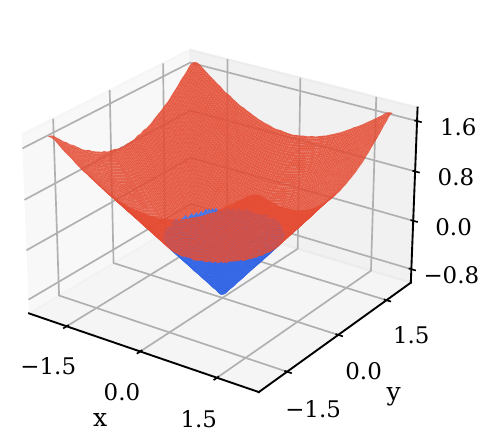} &
\includegraphics[width=0.235\textwidth]{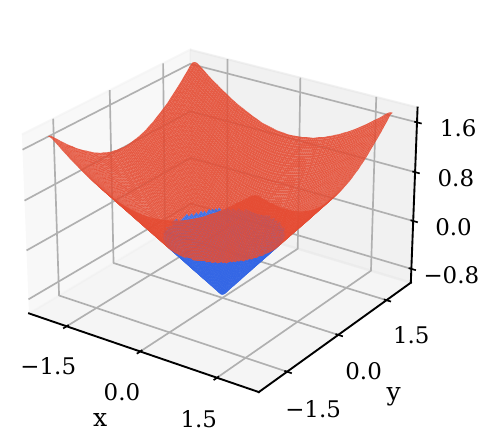} \\
\end{tabular}

\vspace{-1mm}
\caption{\textbf{INR fields trained with different budgets.}
Columns correspond to training steps/epochs: $10^3$, $3{\times}10^3$, $10^4$, $1.5{\times}10^4$.
Top: gradient magnitude $|\nabla\psi_\theta|$ on a uniform grid, indicating improved
near-interface non-degeneracy as training increases.
Bottom: the learned surface $z=\psi_\theta(x,y)$, showing reduced oscillations and a more
regular level-set landscape for longer training.}
\label{fig:inr_fields_epochs}
\end{figure*}

\begin{figure*}[t]
\centering
\setlength{\tabcolsep}{3pt}
\renewcommand{\arraystretch}{1.0}

% -------- Row 1: success masks --------
\begin{tabular}{cccc}
\includegraphics[width=0.235\textwidth]{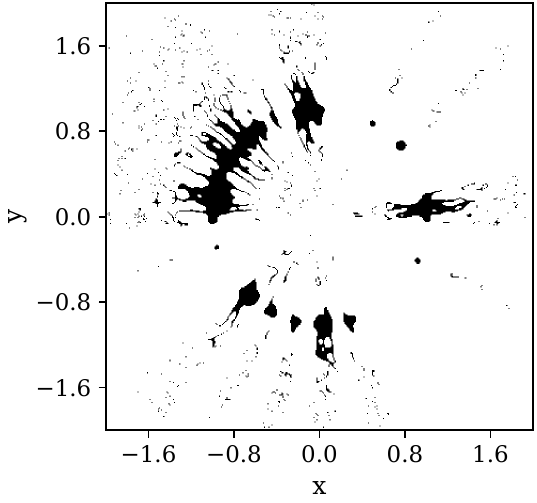} &
\includegraphics[width=0.235\textwidth]{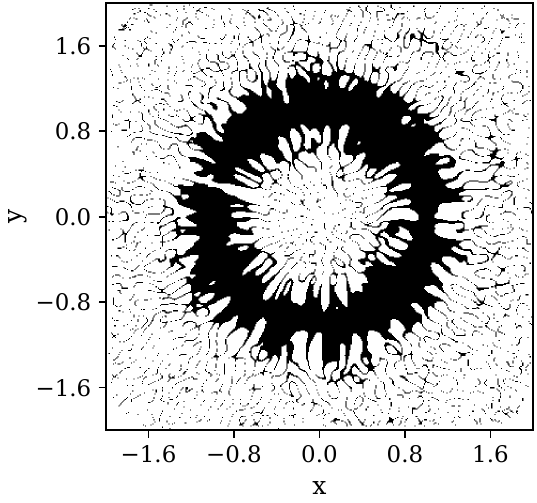} &
\includegraphics[width=0.235\textwidth]{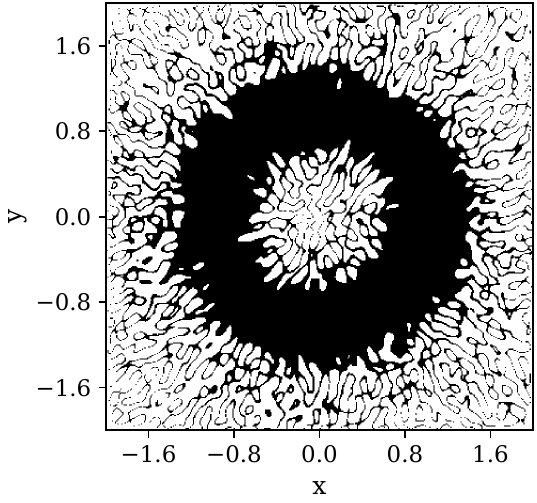} &
\includegraphics[width=0.235\textwidth]{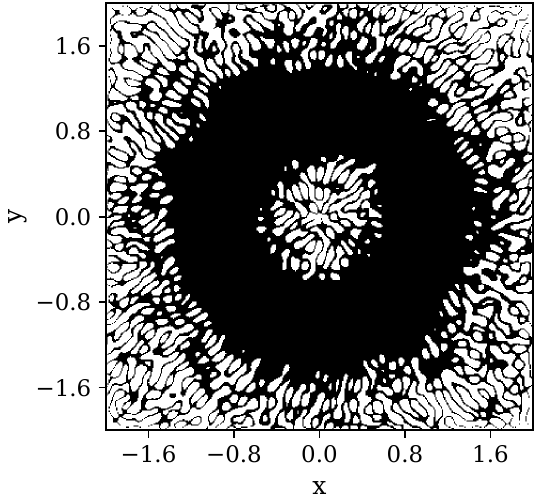} \\
\end{tabular}

\vspace{0.25em}

% shared success colorbar
\includegraphics[width=0.3\textwidth]{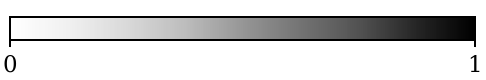}

\vspace{0.6em}

% -------- Row 2: error maps --------
\begin{tabular}{cccc}
\includegraphics[width=0.235\textwidth]{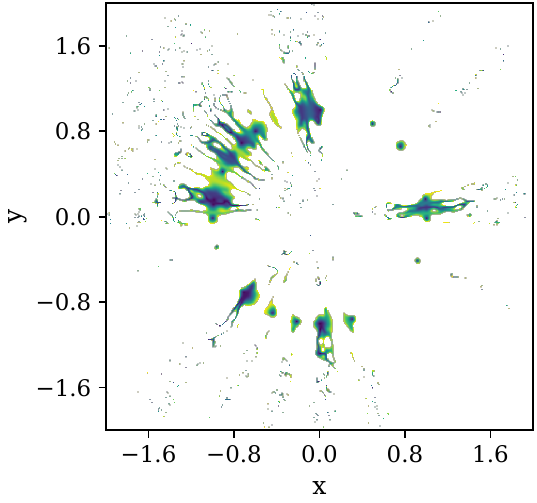} &
\includegraphics[width=0.235\textwidth]{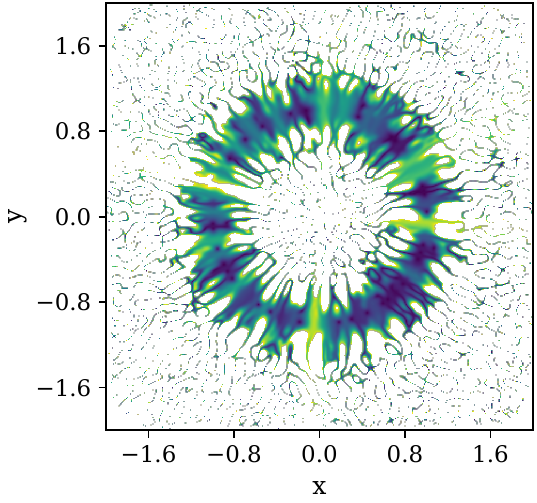} &
\includegraphics[width=0.235\textwidth]{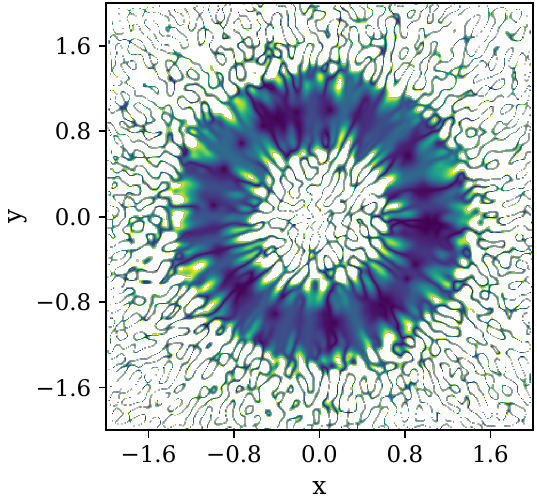} &
\includegraphics[width=0.235\textwidth]{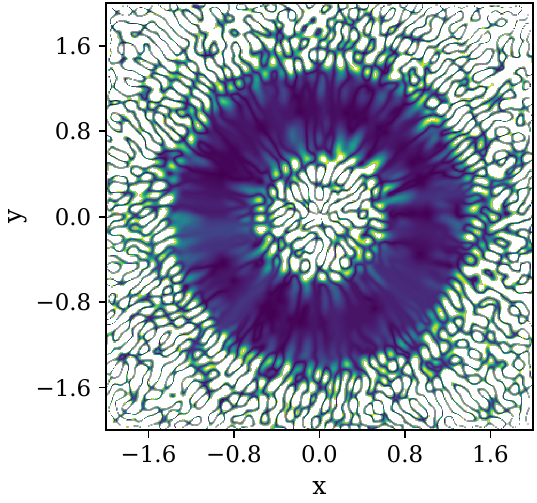} \\
\end{tabular}

\vspace{0.25em}

% shared error colorbar
\includegraphics[width=0.3\textwidth]{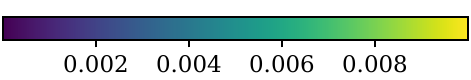}

\vspace{-1mm}
\caption{\textbf{Newton closest-point projection on INR zero level sets across training epochs.}
We apply the Newton map \eqnref{eq:newton_projection_inr} to project grid points onto
$\Gamma^p=\{\psi_\theta=0\}$.
Top: success masks (black indicates convergence to the \emph{true} closest point on the
reference circle, white indicates failure or convergence to an incorrect point).
Bottom: projection error $\|\mathbf{x}_{\mathrm{final}}-\mathbf{x}^*(\mathbf{x}_0)\|$
on the successful set.
Longer training yields larger success regions and smaller errors, consistent with the
role of gradient non-degeneracy in \lemref{lem:projection}.}
\label{fig:inr_projection_epochs}
\end{figure*}

\begin{figure}[t]
\centering
\begin{tikzpicture}
\begin{axis}[
    width=0.5\linewidth,
    height=0.35\linewidth,
    xmode=log,
    ymode=log,
    xlabel={$\,\widehat{\varepsilon}_\infty / \widehat{\tilde c}_0$},
    ylabel={$\,\widehat d_H$},
    grid=both,
    legend style={at={(0.02,0.98)},anchor=north west},
    tick label style={font=\small},
    label style={font=\small},
]

% ---- Data points: (bound, dH_hat) ----
\addplot[
    only marks,
    mark=*,
    mark size=2.2pt
] coordinates {
    (1.0279e-01, 4.7688e-02) % 1k
    (1.3126e-02, 1.0506e-02) % 3k
    (6.2101e-03, 3.4804e-03) % 10k
    (4.2956e-03, 2.0571e-03) % 15k
};
\addlegendentry{Trained INRs}

% ---- Reference line y = x ----
\addplot[
    domain=1e-3:2e-1,
    samples=200,
    thick
] {x};
\addlegendentry{Lemma's Limit}

\node[anchor=south west,font=\scriptsize] at (axis cs:1.0279e-01,4.7688e-02) {1k};
\node[anchor=south west,font=\scriptsize] at (axis cs:1.3126e-02,1.0506e-02) {3k};
\node[anchor=south west,font=\scriptsize] at (axis cs:6.2101e-03,3.4804e-03) {10k};
\node[anchor=south west,font=\scriptsize] at (axis cs:4.2956e-03,2.0571e-03) {15k};

\end{axis}
\end{tikzpicture}
\caption{Empirical validation of \lemref{lem:hausdorff}: the measured Hausdorff
distance \(\widehat d_H\) stays below the predicted upper bound
\(\widehat{\varepsilon}_\infty/\widehat{\tilde c}_0\) (line \(y=x\)) across INRs
trained for different epochs.}
\label{fig:hausdorff_bound_validation}
\end{figure}
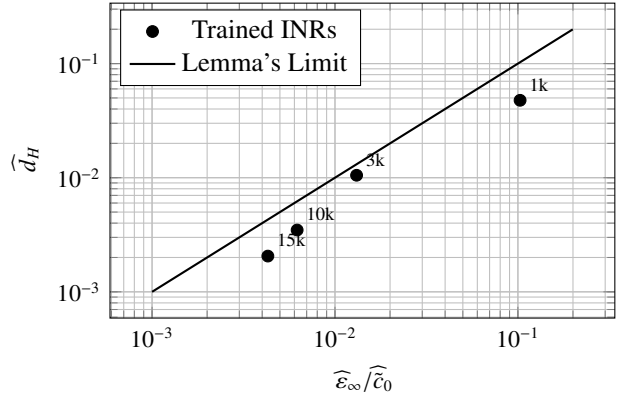

\label{sec:hausdorff_validation}

We next empirically validate the Hausdorff perturbation estimate of
\lemref{lem:hausdorff} on a controlled two-dimensional test case.
We consider the unit-circle signed distance function (SDF)
\begin{equation}
\label{eq:true_sdf_circle}
\phi(x)=\|x\|-1,
\end{equation}
whose zero level set $\Gamma=\{\phi=0\}$ is the unit circle.
We train a SIREN-based implicit neural representation (INR) $\psi_\theta$ to approximate
$\phi$ over a fixed bounding box, using a data term together with an Eikonal penalty in
a narrow band around the interface. We repeat training using four different computational
budgets (increasing number of optimization steps), which we interpret as increasingly
accurate geometric surrogates.

\paragraph{Regularity trends across training budgets.}
~\figref{fig:inr_fields_epochs} visualizes the learned implicit fields for the four
training budgets. The top row shows the gradient magnitude $|\nabla\psi_\theta|$,
while the bottom row shows the surface plot $z=\psi_\theta(x,y)$.
As training proceeds, the level-set landscape becomes smoother and less oscillatory, and
the near-interface gradient magnitude becomes increasingly well-behaved and closer to the
SDF ideal $|\nabla\phi|\equiv 1$.
These observations are consistent with \assmref{ass:geom}, which requires both
(i) non-degeneracy of $\nabla\psi_\theta$ and (ii) bounded curvature surrogates through
a $C^{1,1}$-type control in a tubular neighborhood of $\Gamma^p$.

\paragraph{Projection stability of INR interfaces.}
To further connect these regularity improvements to geometric well-posedness, we apply
the Newton normal-projection iteration
\begin{equation}
\label{eq:newton_projection_inr}
\mathbf{x}_{k+1}
=
\mathbf{x}_k
-
\frac{\psi_\theta(\mathbf{x}_k)}{\|\nabla\psi_\theta(\mathbf{x}_k)\|^2}\,
\nabla\psi_\theta(\mathbf{x}_k),
\end{equation}
which is a standard closest-point construction for implicit surfaces.
~\figref{fig:inr_projection_epochs} reports, for each trained INR, (top) the set of
initial points for which Newton iteration converges to the \emph{true} closest point on
the reference circle and (bottom) the resulting closest-point error on the converged set.
Longer training yields both larger convergence regions and smaller projection errors.
This directly supports the mechanism behind \lemref{lem:projection}: stable normal
projection is sensitive to gradient non-degeneracy, and deteriorates when
$\|\nabla\psi_\theta\|$ becomes small or irregular near the interface.

\paragraph{Empirical verification of the Hausdorff bound.}
We now validate the quantitative estimate in \lemref{lem:hausdorff}.
Let $\Gamma^p=\{\psi_\theta=0\}$ denote the INR-induced boundary.
We evaluate errors in a tubular neighborhood around the true interface,
\[
\mathcal{T}_h=\{x:|\phi(x)|<h\},
\]
and define the empirical uniform mismatch and gradient non-degeneracy as
\begin{equation}
\label{eq:empirical_eps_c0}
\widehat{\varepsilon}_\infty
:=
\max_{x\in\mathcal{T}_h}
|\psi_\theta(x)-\phi(x)|,
\qquad
\widehat{\tilde c}_0
:=
\min_{x\in\mathcal{T}_h}
\|\nabla\psi_\theta(x)\|.
\end{equation}
Since the reference field $\phi$ is an exact SDF, we have $c_0=1$ for the true geometry.
Motivated by \lemref{lem:hausdorff}, we compare the measured Hausdorff discrepancy
$\widehat d_H(\Gamma,\Gamma^p)$ to the predicted upper bound
\begin{equation}
\label{eq:empirical_hausdorff_bound}
\widehat{\mathrm{Bound}}
:=
\frac{\widehat{\varepsilon}_\infty}{\min(1,\widehat{\tilde c}_0)}.
\end{equation}
~\figref{fig:hausdorff_bound_validation} plots $\widehat d_H$ versus
$\widehat{\varepsilon}_\infty/\widehat{\tilde c}_0$ in log--log scale, together with the
reference line $y=x$ corresponding to the bound predicted by \lemref{lem:hausdorff}.
All measured points lie below the line $y=x$, confirming that the empirical behavior of
the trained INRs is consistent with the theoretical estimate.

\paragraph{Interpretation}
The results highlight two key phenomena.
First, the Hausdorff discrepancy decreases systematically as training increases, tracking
the reduction in the uniform mismatch $\widehat{\varepsilon}_\infty$.
Second, the bound exhibits explicit dependence on $\widehat{\tilde c}_0$, emphasizing the
importance of gradient non-degeneracy: even small function mismatch may yield large
geometric error if $\|\nabla\psi_\theta\|$ degenerates near the interface. In our
experiment, $\widehat{\tilde c}_0$ approaches the ideal SDF value $1$ for larger training
budgets, and the predicted upper bound becomes increasingly sharp.
Overall, these results provide concrete evidence that the geometric perturbation estimate
of \lemref{lem:hausdorff} captures the dominant scaling between INR training error and
interface discrepancy for practical neural level-set reconstructions.
\subsection{Empirical Error Scaling under Geometric Perturbations}
\label{sec:error_bounds}
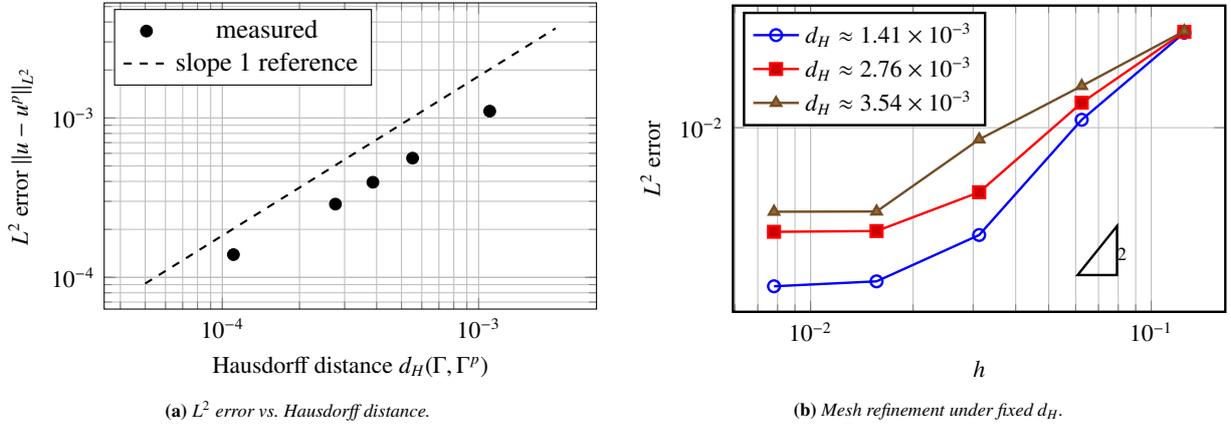
\begin{figure*}[t]
\centering

\begin{subfigure}[t]{0.49\textwidth}
\centering
\begin{tikzpicture}
\begin{axis}[
    width=\linewidth,
    height=0.70\linewidth,
    xmode=log,
    ymode=log,
    grid=both,
    xlabel={Hausdorff distance $d_H(\Gamma,\Gamma^p)$},
    ylabel={$L^2$ error $\|u-u^p\|_{L^2}$},
    legend style={at={(0.02,0.98)},anchor=north west},
    tick label style={font=\small},
    label style={font=\small},
]

% ---- data ----
\addplot[
    only marks,
    mark=*,
    mark size=2.2pt
] coordinates {
    (1.10485e-04, 1.38921e-04)
    (2.76213e-04, 2.88161e-04)
    (3.86698e-04, 3.95160e-04)
    (5.52426e-04, 5.59827e-04)
    (1.10485e-03, 1.10331e-03)
};
\addlegendentry{measured}

% ---- reference slope +1: y = C x ----
\addplot[
    domain=5e-5:2e-3,
    samples=200,
    thick,
    dashed
] { (1.01032e-04/5.52426e-05) * x };
\addlegendentry{slope $1$ reference}

\end{axis}
\end{tikzpicture}
\caption{$L^2$ error vs.\ Hausdorff distance.}
\label{fig:l2_vs_hausdorff_side}
\end{subfigure}
\hfill
\begin{subfigure}[t]{0.49\textwidth}
\centering
\begin{tikzpicture}
\begin{axis}[
    width=\linewidth,
    height=0.70\linewidth,
    xmode=log,
    ymode=log,
    xlabel={$h$},
    ylabel={$L^2\text{ error}$},
    grid=both,
    legend style={at={(0.02,0.98)},anchor=north west,font=\small},
    tick label style={font=\small},
    label style={font=\small},
    line width=0.9pt,
    mark size=2.2pt,
]

% ---------------- data table ----------------
\pgfplotstableread[col sep=comma]{
alpha_dir,alpha,level,h,L2_error,Hausdorff,Hausdorff_mean
alpha0.001953125,0.001953125,8,0.0078125,0.00278998,0.00276213,0.002761758
alpha0.001953125,0.001953125,7,0.015625,0.0028157,0.00276212,0.002761758
alpha0.001953125,0.001953125,6,0.03125,0.00453095,0.00276206,0.002761758
alpha0.001953125,0.001953125,5,0.0625,0.0135777,0.00276179,0.002761758
alpha0.001953125,0.001953125,4,0.125,0.0323833,0.00276069,0.002761758
alpha0.0025,0.0025,8,0.0078125,0.00356916,0.00353553,0.0035350599999999996
alpha0.0025,0.0025,7,0.015625,0.00357819,0.00353552,0.0035350599999999996
alpha0.0025,0.0025,6,0.03125,0.00865869,0.00353544,0.0035350599999999996
alpha0.0025,0.0025,5,0.0625,0.0166679,0.0035351,0.0035350599999999996
alpha0.0025,0.0025,4,0.125,0.0326803,0.00353371,0.0035350599999999996
alpha0.001,0.001,8,0.0078125,0.00143048,0.00141421,0.0014140140000000001
alpha0.001,0.001,7,0.015625,0.00151974,0.00141421,0.0014140140000000001
alpha0.001,0.001,6,0.03125,0.00268244,0.00141417,0.0014140140000000001
alpha0.001,0.001,5,0.0625,0.0110179,0.00141403,0.0014140140000000001
alpha0.001,0.001,4,0.125,0.0320447,0.00141345,0.0014140140000000001
}\datatable

% ---------------- curves (one per Hausdorff distance) ----------------
\addplot+[mark=o]
table[
    x=h,
    y=L2_error,
    col sep=comma,
    restrict expr to domain={\thisrow{alpha}}{0.001:0.001}
] {\datatable};
\addlegendentry{$d_H \approx 1.41\times10^{-3}$}

\addplot+[mark=square*]
table[
    x=h,
    y=L2_error,
    col sep=comma,
    restrict expr to domain={\thisrow{alpha}}{0.001953125:0.001953125}
] {\datatable};
\addlegendentry{$d_H \approx 2.76\times10^{-3}$}

\addplot+[mark=triangle*]
table[
    x=h,
    y=L2_error,
    col sep=comma,
    restrict expr to domain={\thisrow{alpha}}{0.0025:0.0025}
] {\datatable};
\addlegendentry{$d_H \approx 3.54\times10^{-3}$}

% ---------------- slope indicator triangles ----------------
% \coordinate (S1A) at (rel axis cs:0.70,0.25);
% \coordinate (S1B) at (rel axis cs:0.78,0.25);
% \coordinate (S1C) at (rel axis cs:0.78,0.33);
% \draw[black] (S1A) -- (S1B) -- (S1C) -- cycle;
% \node[font=\scriptsize] at (rel axis cs:0.79,0.29) {$1$};

\coordinate (S2A) at (rel axis cs:0.70,0.12);
\coordinate (S2B) at (rel axis cs:0.78,0.12);
\coordinate (S2C) at (rel axis cs:0.78,0.28);
\draw[black] (S2A) -- (S2B) -- (S2C) -- cycle;
\node[font=\scriptsize] at (rel axis cs:0.79,0.19) {$2$};

\end{axis}
\end{tikzpicture}
\caption{Mesh refinement under fixed $d_H$.}
\label{fig:l2_vs_h_multidh_side}
\end{subfigure}
    \vspace{-1mm}
    \caption{\textbf{Solution error under geometric perturbations using an octree background grid with SBM boundary treatment.}
    Left: $L^2$ error as a function of the Hausdorff boundary mismatch.
    Right: $L^2$ convergence with mesh refinement for several fixed perturbation magnitudes.}
    \label{fig:error_bounds_side_by_side}
\end{figure*}
We now empirically examine how geometric boundary perturbations influence the numerical solution of an elliptic boundary value problem in an INR-based unfitted discretization pipeline. Classical domain-perturbation theory for elliptic problems posed on Lipschitz domains predicts that, for sufficiently small perturbations, the $L^2$ solution error scales approximately linearly with the Hausdorff boundary mismatch, while discretization error decays algebraically with mesh refinement until a geometry-induced floor is reached. Our objective in this section is to illustrate how this scaling behavior manifests in practice when (i) geometry is represented implicitly by neural fields and (ii) boundary conditions are imposed weakly using the shifted boundary method (SBM) on an unfitted background grid.

Throughout this section, the reference domain $\Omega$ is the unit disk with boundary $\Gamma=\partial\Omega$. We generate a family of perturbed domains $\Omega^p$ whose boundaries $\Gamma^p$ deviate from $\Gamma$ in a controlled manner. The geometric discrepancy is quantified using the Hausdorff distance $d_H(\Gamma,\Gamma^p)$.

\paragraph{Model problem and manufactured solution.}
We consider the Poisson problem
\begin{equation}
\label{eq:poisson_model}
-\Delta u = f \quad \text{in } \Omega, 
\qquad
u = g \quad \text{on } \Gamma .
\end{equation}
To enable direct error evaluation, we employ the manufactured solution
\begin{equation}
\label{eq:manufactured_solution}
u(x,y)=\sin(\pi x),
\end{equation}
which implies the forcing term
\begin{equation}
\label{eq:forcing_term}
f(x,y) = -\Delta u(x,y) = \pi^2 \sin(\pi x),
\end{equation}
and boundary data $g=u|_{\Gamma}$. For each perturbed geometry $\Omega^p$, we solve
the corresponding perturbed-domain problem
\begin{equation}
\label{eq:poisson_perturbed}
-\Delta u^p = f \quad \text{in } \Omega^p, 
\qquad
u^p = g^p \quad \text{on } \Gamma^p ,
\end{equation}
where the Dirichlet data is taken consistently from the manufactured solution, i.e., $g^p = u|_{\Gamma^p}$. All computations are performed using an unfitted discretization on octree-based background grids, where Dirichlet boundary conditions on $\Gamma^p$ are imposed weakly using the shifted boundary method (SBM).

Since $u$ and $u^p$ are defined on different domains, we compare them on a fixed background domain $D\supset \Omega\cup\Omega^p$ by extension and report
\[
\|u-u^p\|_{L^2(D)} .
\]
This provides a consistent measure across perturbed geometries and isolates the effect of geometric mismatch from domain-dependent normalization.

\subsubsection{Solution error versus Hausdorff perturbation}
\label{sec:l2_vs_hausdorff}

~\figref{fig:error_bounds_side_by_side}(\textit{left}) reports the measured $L^2$ error $\|u-u^p\|_{L^2(D)}$ as a function of the Hausdorff boundary mismatch $d_H(\Gamma,\Gamma^p)$ in log--log scale. The dashed line indicates the linear
scaling predicted by classical perturbation estimates. The measured data exhibits an approximately affine trend with slope close to $+1$, indicating the empirical scaling
\begin{equation}
\label{eq:l2_scaling_observed}
\|u-u^p\|_{L^2(D)} \;\approx\; C\, d_H(\Gamma,\Gamma^p),
\end{equation}
for the perturbation magnitudes considered here. This observation is consistent with theoretical expectations for elliptic problems on perturbed Lipschitz domains and suggests that INR-defined surrogate geometries, when used with SBM, inherit the same first-order sensitivity to boundary mismatch as classical fitted discretizations.

\subsection{Interplay of discretization and geometry error}
\label{sec:l2_vs_h}

We next investigate how mesh refinement interacts with fixed geometric perturbations. For each selected perturbation level (equivalently, for each fixed Hausdorff distance $d_H(\Gamma,\Gamma^p)$), we solve \eqnref{eq:poisson_perturbed} on a sequence of
uniformly refined octree grids with characteristic mesh size $h$. \figref{fig:error_bounds_side_by_side}(\textit{right}) shows the convergence of $\|u-u^p\|_{L^2(D)}$ with decreasing $h$ for three representative perturbation magnitudes. For coarse grids, the total error decreases under refinement, reflecting the reduction in discretization error. However, once $h$ becomes sufficiently small, the curves plateau and no longer improve with refinement, indicating the onset of a geometry-limited regime. This behavior is consistent with an additive error decomposition of the form
\[
\text{total error}
\;\approx\;
\underbrace{\text{geometry-induced error}}_{\text{set by } d_H(\Gamma,\Gamma^p)}
\;+\;
\underbrace{\text{discretization error}}_{\to 0\ \text{as}\ h\to 0}.
\]
The slope markers in the pre-asymptotic regime indicate algebraic decay with mesh refinement, while the plateau confirms that boundary mismatch acts as an irreducible error floor unless $d_H(\Gamma,\Gamma^p)$ is reduced in tandem with the discretization scale. 

Taken together, the two plots in ~\figref{fig:error_bounds_side_by_side} demonstrate that accurate elliptic solutions on INR-defined surrogate geometries require simultaneous control of both the mesh resolution $h$ and the geometric discrepancy $d_H(\Gamma,\Gamma^p)$. In the context of  implicit neural representations, this highlights the practical importance of aligning training tolerances with the intended discretization accuracy in order to avoid geometry-dominated saturation.

\section{Conclusion}
\label{sec:conclusion}

We have developed a unified theoretical framework that links the approximation quality
of Implicit Neural Representations (INRs) to the stability and accuracy of elliptic PDE
solutions posed on INR-defined domains. While INRs provide a compact, continuous, and
differentiable representation of geometry, our results show that reconstruction-oriented
metrics alone (e.g., Chamfer or Hausdorff distance between extracted surfaces) do not
fully capture whether a learned field is suitable for scientific simulation. In
particular, simulation readiness depends critically on the \emph{differential
regularity} of the implicit field in a neighborhood of the interface.

Our main contributions can be summarized as follows.
\begin{enumerate}
\item \textbf{Minimal geometric regularity for well-posed projection.}
We identified sufficient regularity conditions on the implicit field in a tubular neighborhood of the interface that guarantee existence and uniqueness of the normal projection map and prevent pathological geometric behavior. Concretely, we require
$\psi\in C^{1,1}$ locally, together with a uniform gradient non-degeneracy bound
$\|\nabla\psi\|\ge c_0>0$ and bounded curvature surrogate $\|\nabla^2\psi\|_{\op}\le C_\psi$.
These constants appear explicitly in the stability estimates, clarifying how
vanishing gradients can amplify geometric errors and destabilize projection-based
constructions.

\item \textbf{Training error $\Rightarrow$ Hausdorff boundary perturbation.}
Assuming a uniform approximation tolerance $\|\psi-\phi\|_{L^\infty(\mathcal{T}_h)}
\le \varepsilon_\infty$, we derived bounds relating the learned boundary
$\Gamma^p=\{\psi=0\}$ to the reference boundary $\Gamma=\{\phi=0\}$ through
Hausdorff distance estimates of the form
$d_H(\Gamma,\Gamma^p)\lesssim \varepsilon_\infty/c_0$,
making explicit the role of the gradient lower bound.

\item \textbf{Geometry-to-solution stability and training guidelines.}
Combining geometric perturbation estimates with sharp stability results for elliptic
boundary value problems, we obtained explicit a priori bounds showing how geometric
perturbations propagate into PDE solution error. Together with standard finite element
approximation theory, this yields practical accuracy requirements for INR training:
to avoid geometry-dominated saturation under mesh refinement, the INR accuracy must be
tightened commensurately with the target discretization scale. In particular, for
first-order methods this leads to a simple rule-of-thumb: the geometry error must decay
faster than the discretization error in order to preserve asymptotic convergence.
\end{enumerate}

\paragraph{Future directions}
Several extensions are natural. First, the explicit separation between geometric and
discretization error motivates adaptive strategies that couple INR training tolerance
to discretization resolution and error indicators, avoiding both under-training (which
limits accuracy) and over-training (which wastes computation). Second, extending the
analysis to higher-order discretizations and higher-regularity PDE settings will
require sharper control of geometry through stronger smoothness assumptions on the
implicit field. Finally, generalizing these stability mechanisms to time-dependent
domains and moving boundary problems would provide a theoretical foundation for INR-based
simulation pipelines in fluid--structure interaction and free-boundary dynamics.

Overall, our results provide a rigorous framework for using INRs as a reliable geometric primitive in high-fidelity scientific computing, clarifying which regularity and training requirements are necessary for stable, convergent PDE simulation.

\bibliographystyle{elsarticle-harv}
\bibliography{main}

@article{savare2002domain,
  title={Domain perturbations and estimates for the solutions of second order elliptic equations},
  author={Savar{\'e}, Giuseppe and Schimperna, Giulio},
  journal={Journal de math{\'e}matiques pures et appliqu{\'e}es},
  volume={81},
  number={11},
  pages={1071--1112},
  year={2002},
  publisher={Elsevier}
}

@article{jerison1995,
  author = {Jerison, David and Kenig, Carlos E.},
  title = {The inhomogeneous {D}irichlet problem in {L}ipschitz domains},
  journal = {Journal of Functional Analysis},
  volume = {130},
  number = {1},
  pages = {161--219},
  year = {1995},
  doi = {10.1006/jfan.1995.1067}
}

@article{grisvard1985elliptic,
  title={"Elliptic Problems in Nonsmooth Domains"},
  author={Grisvard, P},
  journal={Pitman},
  year={1985}
}

@article{henrot1994continuity,
  title={Continuity with respect to the domain for the Laplacian: a survey},
  author={Henrot, Antoine},
  journal={Control and Cybernetics},
  volume={23},
  number={3},
  pages={427--443},
  year={1994}
}

@article{ciarlet1972interpolation,
  title={Interpolation theory over curved elements, with applications to finite element methods},
  author={Ciarlet, Philippe G and Raviart, P-A},
  journal={Computer Methods in Applied Mechanics and Engineering},
  volume={1},
  number={2},
  pages={217--249},
  year={1972},
  publisher={Elsevier}
}

@inproceedings{sitzmann2020implicit,
 author = {Sitzmann, Vincent and Martel, Julien and Bergman, Alexander and Lindell, David and Wetzstein, Gordon},
 booktitle = {Advances in Neural Information Processing Systems},
 editor = {H. Larochelle and M. Ranzato and R. Hadsell and M.F. Balcan and H. Lin},
 pages = {7462--7473},
 publisher = {Curran Associates, Inc.},
 title = {Implicit Neural Representations with Periodic Activation Functions},
 url = {https://proceedings.neurips.cc/paper_files/paper/2020/file/53c04118df112c13a8c34b38343b9c10-Paper.pdf},
 volume = {33},
 year = {2020},
address = {virtual}
}

@inproceedings{park2019deepsdf,
    author = {Park, Jeong Joon and Florence, Peter and Straub, Julian and Newcombe, Richard and Lovegrove, Steven},
    title = {{DeepSDF}: Learning Continuous Signed Distance Functions for Shape Representation},
    booktitle = {Proceedings of the IEEE/CVF Conference on Computer Vision and Pattern Recognition (CVPR)},
    month = {June},
    year = {2019},
    publisher = {IEEE},
    address = {Long Beach},
    pages="1--12",
}

@inproceedings{chibane2020implicit,
  title={Implicit functions in feature space for {3D} shape reconstruction and completion},
  author={Chibane, Julian and Alldieck, Thiemo and Pons-Moll, Gerard},
  booktitle = {Proceedings of the IEEE/CVF Conference on Computer Vision and Pattern Recognition (CVPR)},
  pages={6970--6981},
  year={2020},
publisher = {IEEE},
address = {virtual}
}

@inproceedings{gropp2020implicit,
    author = {Gropp, Amos and Yariv, Lior and Haim, Niv and Atzmon, Matan and Lipman, Yaron},
    title = {Implicit geometric regularization for learning shapes},
    year = {2020},
    publisher = {JMLR.org},
    booktitle = {Proceedings of the 37th International Conference on Machine Learning},
    articleno = {355},
    numpages = {11},
    series = {ICML'20},
    address = {virtual},
    pages = {1--12}
}

@Article{hsu2016direct,
  author  = {Hsu, Ming-Chen and Wang, Chenglong and Xu, Fei and Herrema, Austin J and Krishnamurthy, Adarsh},
  journal = {Computer Aided Geometric Design},
  title   = {Direct immersogeometric fluid flow analysis using {B-rep} {CAD} models},
  year    = {2016},
  pages   = {143--158},
  volume  = {43},
}

@Article{main2018shifted1,
  author  = {Main, Alex and Scovazzi, Guglielmo},
  journal = {Journal of Computational Physics},
  title   = {The shifted boundary method for embedded domain computations. {Part I}: {Poisson} and {Stokes} problems},
  year    = {2018},
  pages   = {972--995},
  volume  = {372},
}

@Article{main2018shifted2,
  author  = {Main, Alex and Scovazzi, Guglielmo},
  journal = {Journal of Computational Physics},
  title   = {The shifted boundary method for embedded domain computations. {Part II}: {Linear} advection--diffusion and incompressible {Navier--Stokes} equations},
  year    = {2018},
  pages   = {996--1026},
  volume  = {372},
}

@article{JAISWAL2024117426,
title = {Mesh-driven resampling and regularization for robust point cloud-based flow analysis directly on scanned objects},
journal = {Computer Methods in Applied Mechanics and Engineering},
volume = {432},
pages = {117426},
year = {2024},
issn = {0045-7825},
doi = {https://doi.org/10.1016/j.cma.2024.117426},
url = {https://www.sciencedirect.com/science/article/pii/S0045782524006819},
author = {Monu Jaiswal and Ashton M. Corpuz and Ming-Chen Hsu}
}

@Article{karki2025direct,
  title={Direct flow simulations with implicit neural representation of complex geometry},
  author={Karki, Samundra and Shadkhah, Mehdi and Yang, Cheng-Hau and Balu, Aditya and Scovazzi, Guglielmo and Krishnamurthy, Adarsh and Ganapathysubramanian, Baskar},
  journal={Computer Methods in Applied Mechanics and Engineering},
  volume={446},
  pages={118248},
  year={2025},
  publisher={Elsevier}
}

@article{karki2025mechanics,
title = {Mechanics simulation with Implicit Neural Representations of complex geometries},
journal = {Computer-Aided Design},
volume = {190},
pages = {103978},
year = {2026},
issn = {0010-4485},
doi = {https://doi.org/10.1016/j.cad.2025.103978},
url = {https://www.sciencedirect.com/science/article/pii/S0010448525001393},
author = {Samundra Karki and Ming-Chen Hsu and Adarsh Krishnamurthy and Baskar Ganapathysubramanian},
}

@Article{nguyen2015isogeometric,
  author  = {Nguyen, Vinh Phu and Anitescu, Cosmin and Bordas, St{\'e}phane PA and Rabczuk, Timon},
  journal = {Mathematics and Computers in Simulation},
  title   = {Isogeometric analysis: An overview and computer implementation aspects},
  year    = {2015},
  pages   = {89--116},
  volume  = {117},
}

@Article{liu2022eighty,
  author  = {Liu, Wing Kam and Li, Shaofan and Park, Harold S},
  journal = {Archives of Computational Methods in Engineering},
  title   = {Eighty years of the finite element method: birth, evolution, and future},
  year    = {2022},
  number  = {6},
  pages   = {4431--4453},
  volume  = {29},
}

@Article{peskin1972flow,
  author  = {Peskin, Charles S},
  journal = {Journal of Computational Physics},
  title   = {Flow patterns around heart valves: {A} numerical method},
  year    = {1972},
  number  = {2},
  pages   = {252--271},
  volume  = {10},
}

@Article{CHIBA1998145,
  author  = {Chiba, N. and Nishigaki, I. and Yamashita, Y. and Takizawa, C. and Fujishiro, K.},
  journal = {Computer Methods in Applied Mechanics and Engineering},
  title   = {A flexible automatic hexahedral mesh generation by boundary-fit method},
  year    = {1998},
  issn    = {0045-7825},
  number  = {1},
  pages   = {145--154},
  volume  = {161},
}

@TechReport{mchenry2008overview,
  author      = {McHenry, Kenton and Bajcsy, Peter},
  institution = {National Center for Supercomputing Applications, University of Illinois at Urbana-Champaign},
  title       = {An overview of {3D} data content, file formats and viewers},
  year        = {2008},
  address     = {1205 W Clark St, Urbana, IL 61801},
  month       = {October},
  note        = {Technical Report, Image Spatial Data Analysis Group},
  number      = {ISDA08-002},
}

@InProceedings{wang2021neus,
  author    = {Wang, Peng and Liu, Lingjie and Liu, Yuan and Theobalt, Christian and Komura, Taku and Wang, Wenping},
  booktitle = {Proceedings of the 35th International Conference on Neural Information Processing Systems},
  title     = {{NeuS}: {Learning} neural implicit surfaces by volume rendering for multi-view reconstruction},
  year      = {2021},
  pages     = {27171--27183},
}

@phdthesis{persson2005mesh,
  title={Mesh generation for implicit geometries},
  author={Persson, Per-Olof},
  year={2005},
  school={Massachusetts Institute of Technology}
}
\clearpage
\appendix
\section{Proof of Existence of Projection}
\label{app:projection}
In this section, we provide a complete proof for \lemref{lem:projection}.
\begin{proof}
\textbf{Step 1: Regularity of $\Gamma$ and the unit normal.}
By Assumption~\ref{ass:geom}, $\psi\in C^{1,1}(\mathcal T_h)$ and
$\|\nabla\psi(x)\|\ge c_0>0$ on $\mathcal T_h$.
Hence the zero level set
\[
\Gamma := \{x\in\mathbb R^n : \psi(x)=0\}
\]
is a $C^{1,1}$ embedded hypersurface. Define the unit normal field
\[
\mathbf n(x) := \frac{\nabla\psi(x)}{\|\nabla\psi(x)\|},\qquad x\in\Gamma.
\]
For $x,y\in\Gamma$, using
$\big\|\frac{a}{\|a\|}-\frac{b}{\|b\|}\big\|\le \frac{2}{\min(\|a\|,\|b\|)}\|a-b\|$
and $\|\nabla\psi\|\ge c_0$ on $\Gamma\subset\mathcal T_h$, we obtain
\[
\|\mathbf n(x)-\mathbf n(y)\|
\le \frac{2}{c_0}\,\|\nabla\psi(x)-\nabla\psi(y)\|
\le \frac{2C_\psi}{c_0}\,\|x-y\|,
\]
since $\nabla\psi$ is Lipschitz on $\mathcal T_h$ with Lipschitz constant bounded by $C_\psi$.
Thus $\mathbf n$ is Lipschitz on $\Gamma$ with
\[
\Lip(\mathbf n)\le L_n := \frac{2C_\psi}{c_0}.
\]

\medskip
\textbf{Step 2: Existence of a reach neighborhood and uniqueness of the metric projection.}
Since $\Gamma$ is a $C^{1,1}$ hypersurface with Lipschitz unit normal, it has
\emph{positive reach}. Consequently, there exists $\rho>0$, depending only on $L_n$,
such that the nearest-point projection onto $\Gamma$ is well defined and unique on
\[
U_\rho := \{x\in\mathbb R^n : \dist(x,\Gamma)<\rho\}.
\]
In particular, for every $x\in U_\rho$ there exists a unique $x^*=\Pi_\Gamma(x)\in\Gamma$
satisfying $|x-x^*|=\dist(x,\Gamma)$.

\medskip
\textbf{Step 3: Embedding the level-set tube into the reach neighborhood.}
We show that for $h$ sufficiently small, $\mathcal T_h\subset U_\rho$.
Fix $x\in\mathcal T_h$ and consider the ODE
\begin{equation}\label{eq:flow_ode_full}
\dot\gamma(t)=-\frac{\nabla\psi(\gamma(t))}{\|\nabla\psi(\gamma(t))\|^2},
\qquad \gamma(0)=x.
\end{equation}
Because $\psi\in C^{1,1}(\mathcal T_h)$ and $\|\nabla\psi\|\ge c_0>0$ on $\mathcal T_h$,
the vector field in \eqnref{eq:flow_ode_full} is locally Lipschitz on $\mathcal T_h$ and
the solution exists as long as $\gamma(t)\in\mathcal T_h$.
Along this trajectory,
\[
\frac{d}{dt}\psi(\gamma(t))
=\nabla\psi(\gamma(t))\cdot\dot\gamma(t)
=-1,
\]
hence $\psi(\gamma(t))=\psi(x)-t$.
Let $T:=|\psi(x)|$. For $t\in[0,T]$ we have
$|\psi(\gamma(t))|=|\psi(x)-t|\le |\psi(x)|<h$, so $\gamma(t)\in\mathcal T_h$ on $[0,T]$ and
\[
\psi(\gamma(T))=\psi(x)-T=0,
\qquad\text{so}\qquad \gamma(T)\in\Gamma.
\]
Therefore,
\[
\dist(x,\Gamma)\le \int_0^{T}\|\dot\gamma(t)\|\,dt
=\int_0^{|\psi(x)|}\frac{1}{\|\nabla\psi(\gamma(t))\|}\,dt
\le \frac{|\psi(x)|}{c_0}
< \frac{h}{c_0}.
\]
Choosing
\[
h_{\mathrm{crit}}:=c_0\,\rho,
\]
it follows that whenever $h<h_{\mathrm{crit}}$ we have $\dist(x,\Gamma)<\rho$, i.e.\ $x\in U_\rho$.
Since $x\in\mathcal T_h$ was arbitrary, this proves $\mathcal T_h\subset U_\rho$.

\medskip
\textbf{Step 4: Normal-coordinate representation and uniqueness.}
Fix $h<h_{\mathrm{crit}}$ and let $x\in\mathcal T_h$.
Then $x\in U_\rho$, so the unique nearest-point projection $x^*=\Pi_\Gamma(x)\in\Gamma$ exists.
For a $C^{1,1}$ hypersurface, the first-order optimality condition for the distance minimization
implies
\[
x-x^*\perp T_{x^*}\Gamma,
\]
hence $x-x^*$ is parallel to $\mathbf n(x^*)$.
Define the signed distance
\[
d := (x-x^*)\cdot \mathbf n(x^*).
\]
Then $x-x^*=d\,\mathbf n(x^*)$, and therefore
\[
x = x^* + d\,\mathbf n(x^*),
\]
which yields \eqnref{eq:projection_relation}. Uniqueness of $x^*$ implies uniqueness of $d$.

\medskip
This completes the proof.
\end{proof}

\section{Proof of Hausdorff Bound}
In this section, we provide a complete proof of \lemref{lem:hausdorff}.
\label{app:hausdorff}

\begin{proof}
Let $\Gamma=\{\phi=0\}$ and $\Gamma^p=\{\psi=0\}$, and assume:
(i) $\phi,\psi\in C^{1,1}(\mathcal T_h)$ for a common tubular neighborhood $\mathcal T_h$;
(ii) the nondegeneracy and curvature bounds
\[
\|\nabla\phi\|\ge c_0,\quad \|\nabla^2\phi\|_{\op}\le C_\phi,
\qquad
\|\nabla\psi\|\ge \tilde c_0,\quad \|\nabla^2\psi\|_{\op}\le \tilde C_\psi
\quad \text{on }\mathcal T_h;
\]
(iii) the uniform field error $\|\psi-\phi\|_{L^\infty(\mathcal T_h)}\le\varepsilon_\infty$.

We prove bounds on the two directed distances and then take the maximum.

\textbf{Step 1: From $\Gamma$ to $\Gamma^p$}
Fix $x\in\Gamma$ so $\phi(x)=0$. Then
\begin{equation}\label{eq:psi_on_Gamma_app}
|\psi(x)|=|\psi(x)-\phi(x)|\le\varepsilon_\infty.
\end{equation}
Define
\[
n_\psi(x):=\frac{\nabla\psi(x)}{\|\nabla\psi(x)\|},
\qquad
\sigma:=\operatorname{sign}(\psi(x))\in\{-1,0,1\}.
\]
If $\sigma=0$ then $x\in\Gamma^p$ and $\dist(x,\Gamma^p)=0$, so assume $\sigma=\pm1$.
Consider the ray
\[
r(s)=x-\sigma s\,n_\psi(x),
\qquad
g(s):=\sigma\,\psi(r(s)).
\]
Then $g(0)=|\psi(x)|>0$.

By the chain rule,
\[
g'(s)
=-\,\nabla\psi(r(s))\cdot n_\psi(x).
\]
Using add--subtract and Cauchy--Schwarz,
\[
\nabla\psi(r(s))\cdot n_\psi(x)
=\nabla\psi(x)\cdot n_\psi(x)
+(\nabla\psi(r(s))-\nabla\psi(x))\cdot n_\psi(x)
\ge \|\nabla\psi(x)\|-\|\nabla\psi(r(s))-\nabla\psi(x)\|.
\]
Since $\nabla\psi$ is Lipschitz with constant $\tilde C_\psi$ on $\mathcal T_h$,
\[
\|\nabla\psi(r(s))-\nabla\psi(x)\|
\le \tilde C_\psi\|r(s)-x\|
= \tilde C_\psi s,
\quad \text{as long as } r(s)\in\mathcal T_h.
\]
Hence,
\[
g'(s)\le -\|\nabla\psi(x)\|+\tilde C_\psi s,
\qquad \text{while } r(s)\in\mathcal T_h.
\]
Integrating from $0$ to $s$ yields
\begin{equation}\label{eq:g_majorant_app}
g(s)\le |\psi(x)|-\|\nabla\psi(x)\|\,s+\tfrac12\tilde C_\psi s^2.
\end{equation}

Let $s_{-x}$ denote the smaller nonnegative root of the quadratic on the right:
\[
s_{-x}
=\frac{\|\nabla\psi(x)\|-\sqrt{\|\nabla\psi(x)\|^2-2\tilde C_\psi|\psi(x)|}}{\tilde C_\psi},
\]
which is well-defined provided
\[
\|\nabla\psi(x)\|^2\ge 2\tilde C_\psi|\psi(x)|.
\]
A sufficient condition is
\begin{equation}\label{eq:smallness_disc_app}
\varepsilon_\infty \le \frac{\tilde c_0^2}{2\tilde C_\psi}.
\end{equation}

To ensure the segment $\{r(s):0\le s\le s_{-x}\}$ remains in $\mathcal T_h$, it suffices to require $s_{-x}\le h$.
Using the inequality $1-\sqrt{1-t}\le t$ for $t\in[0,1]$, we obtain
\[
s_{-x}\le \frac{2|\psi(x)|}{\|\nabla\psi(x)\|},
\]
so a sufficient condition is
\begin{equation}\label{eq:smallness_tube_app}
\varepsilon_\infty \le \frac{\tilde c_0}{2}\,h.
\end{equation}

Under \eqnref{eq:smallness_disc_app}--\eqnref{eq:smallness_tube_app},
\eqnref{eq:g_majorant_app} implies $g(s_{-x})\le 0$ while $g(0)>0$.
By continuity, there exists $\bar s\in(0,s_{-x}]$ such that $g(\bar s)=0$,
i.e.\ $\psi(r(\bar s))=0$ and $r(\bar s)\in\Gamma^p$.
Therefore,
\[
\dist(x,\Gamma^p)\le \bar s \le s_{-x}.
\]

Using $\|\nabla\psi(x)\|\ge \tilde c_0$ and $|\psi(x)|\le \varepsilon_\infty$ yields
\begin{equation}\label{eq:Gamma_to_Gammap_bound_app}
\sup_{x\in\Gamma}\dist(x,\Gamma^p)
\le \frac{\varepsilon_\infty}{\tilde c_0}
+\frac{\tilde C_\psi}{2\tilde c_0^{\,3}}\varepsilon_\infty^2
+\mathcal O(\varepsilon_\infty^3).
\end{equation}

\paragraph{Step 2: From $\Gamma^p$ to $\Gamma$}
Fix $y\in\Gamma^p$ so $\psi(y)=0$. Then $|\phi(y)|\le \varepsilon_\infty$.
Repeating Step~1 with $\phi$ in place of $\psi$ (and constants $c_0,C_\phi$),
assuming
\[
\varepsilon_\infty \le \frac{c_0^2}{2C_\phi},
\qquad
\varepsilon_\infty \le \frac{c_0}{2}\,h,
\]
we obtain
\begin{equation}\label{eq:Gammap_to_Gamma_bound_app}
\sup_{y\in\Gamma^p}\dist(y,\Gamma)
\le \frac{\varepsilon_\infty}{c_0}
+\frac{C_\phi}{2c_0^{\,3}}\varepsilon_\infty^2
+\mathcal O(\varepsilon_\infty^3).
\end{equation}

\textbf{Step 3: Hausdorff distance}
By definition,
\[
d_H(\Gamma,\Gamma^p)
=\max\Big\{
\sup_{x\in\Gamma}\dist(x,\Gamma^p),\;
\sup_{y\in\Gamma^p}\dist(y,\Gamma)
\Big\}.
\]
Combining
\eqnref{eq:Gamma_to_Gammap_bound_app}
and
\eqnref{eq:Gammap_to_Gamma_bound_app}
yields
\[
d_H(\Gamma,\Gamma^p)
\le \frac{\varepsilon_\infty}{\min(c_0,\tilde c_0)}
+\mathcal O(\varepsilon_\infty^2),
\]
where the $\mathcal O(\varepsilon_\infty^2)$ constant depends only on
$(c_0,C_\phi)$ and $(\tilde c_0,\tilde C_\psi)$.
\end{proof}

\end{document}